\documentclass[10pt]{article}

\usepackage[normalem]{ulem}
\usepackage{mathtools}

\usepackage{amsmath}
\usepackage{amssymb}
\usepackage{amsfonts}
\usepackage{amsthm}
\usepackage{a4wide}
\usepackage{graphicx}
\usepackage{color}
\usepackage[sans]{dsfont}
\usepackage{pgfkeys}
\usepackage{longtable}
\usepackage{pdflscape}
\usepackage{float}
\usepackage{cancel}
\usepackage{dsfont}
\usepackage{setspace}
\usepackage{array}
\usepackage{hyperref}
\usepackage[english]{babel}
\usepackage[T1]{fontenc}
\usepackage{tikz,pgflibraryshapes}
\usepackage{enumerate}
\usetikzlibrary{decorations}
\usetikzlibrary{decorations.pathmorphing}
    \usetikzlibrary{decorations.pathreplacing}
    \usetikzlibrary{decorations.shapes}
    \usetikzlibrary{decorations.text}
    \usetikzlibrary{decorations.markings}
    \usetikzlibrary{decorations.fractals}
    \usetikzlibrary{decorations.footprints}

\usepackage{wasysym}
\usepackage{authblk}
\usepackage{calrsfs}

\newcommand{\tr}{\mathrm{Tr}}
\newcommand{\id}{\mathds{1}}

\theoremstyle{plain}

\newtheorem{theorem}{Theorem}[subsection]
\newtheorem{lemma}[theorem]{Lemma}
\newtheorem{example}[theorem]{Example}
\newtheorem{assumption}[theorem]{Assumption}

\usepackage{fancyhdr}

\DeclareMathAlphabet{\pazocal}{OMS}{zplm}{m}{n}

\numberwithin{equation}{section}

  \title{Indirect retrieval of information and the emergence of facts in quantum mechanics}
\begin{document}

\author[1]{ \small M. Ballesteros \thanks{miguel.ballesteros@iimas.unam.mx}}
\author[2]{M. Fraas\thanks{fraas@math.lmu.de}}
\author[3]{J. Fr\"ohlich\thanks{juerg@phys.ethz.ch}}
\author[4]{B.Schubnel\thanks{baptiste.schubnel@univ-lorraine.fr}}
\affil[1]{Department of Mathematical Physics, Applied Mathematics and Systems Research Institute (IIMAS),  National Autonomous University of Mexico (UNAM)}
\affil[2]{Mathematisches Institut der Universit\"at M\"unchen,  University of Munich}
\affil[3]{Institut f{\"u}r Theoretische Physik,   ETH Zurich}
\affil[4]{IECL, Universit\'{e} de Lorraine, Metz}

\renewcommand\Authands{ and }

\normalsize 
\date \today

\maketitle
\begin{abstract}

Long sequences of successive direct (projective) measurements or observations of a few
``uninteresting'' physical quantities of a quantum system may reveal indirect, but precise and unambiguous 
information on the values of some very ``interesting'' observables of the system. In this paper, the  
mathematics underlying this claim is developed; i.e., we attempt to contribute to a mathematical theory 
of indirect and, in particular, non-demolition measurements in quantum mechanics. 
Our attempt leads us to make novel uses of classical notions and results of probability theory, such as the ``algebra of functions measurable at infinity'', the Central Limit Theorem, results concerning relative entropy and its role in the theory of large deviations, etc.
\end{abstract}

\section{Introduction: Purpose and scope of paper}\label{intro}
This paper is devoted to a study of the theory of indirect measurements and observations in quantum mechanics, \cite{Kraus}. Our main aim is to develop a general perspective on questions that, in very general terms, can be formulated as follows: What sort of information on a quantum system, S, can be extracted from long, time-ordered sequences of direct (so-called projective) measurements or observations of a $single$ physical quantity of $S$ represented by a self-adjoint operator, $X$; (or of a finite number of physical quantities, $\vec{X} := \lbrace X_1,...,X_r \rbrace$, of S)? What do such sequences of direct measurements/observations reveal about \textit{facts} concerning $S$? What does quantum mechanics tell us about the time evolution of quantum systems subjected to repeated measurements/observations? These are fairly fundamental questions about quantum theory that call for general answers arrived at in an exploration of general concepts and structure that goes beyond the analysis of special examples. The analysis presented in this paper is inspired, to some extent, by the work of experimental groups, see, in particular  \cite{guerlin}, and by various theoretical papers, see \cite{Kraus}, \cite{Mass}, \cite{BaBe}, \cite{BaBeTi}, \cite{BePel}, and references given there. Although our main interest concerns  conceptual aspects of quantum mechanics, the main effort that has gone into this paper concerns the study of simple models (see, e.g.,  \cite{BaBe}) describing a concrete physical situation. This has the advantage that it renders our paper comprehensible for readers who are not used to abstraction. Yet, our analysis is intended to illustrate some -- we believe novel -- $general$ insights: The mathematics underlying the theory of long, time-ordered sequences of direct measurements of some physical quantities of a quantum system $S$ turns out to be closely related to the one underlying the theory of equilibrium (Gibbs) states of long chains of classical lattice gases (or classical spins), with time of the quantum system corresponding to $1D$ physical space of the lattice gas and sequences of outcomes of direct quantum-mechanical measurements corresponding to configurations of particles in the lattice gas. Results of statistical mechanics, such as the equivalence of different ensembles in the thermodynamic limit, or the disjointness of equilibrium states corresponding to different values of some thermodynamic parameters, such as the temperature or a chemical potential, have close cousins in the quantum theory of repeated measurements. For example, the observation of an event of finite duration in a quantum system corresponds to the appearance of short-range order in some bounded spatial region of the lattice gas, and an ``eternal property'' of a quantum system (i.e., a property of $S$ observed in a non-demolition measurement) corresponds to a specific value of an \textit{order parameter} labeling different equilibrium states of the lattice gas.
Mathematical methods, such as large-deviation theory, concentration-of-measure estimates, hypothesis testing, etc., can be transferred from statistical mechanics to quantum theory. \\

In this paper, we will study a quantum system, $S$, with the following properties:
\begin{enumerate}
\item For all practical purposes, $S$ can be considered to be an isolated (``closed'') system. This means that, in the Heisenberg picture, the equations of motion of operators representing physical quantities of $S$ take the form of \textit{Heisenberg equations of motion} in which the ``Hamiltonian'' of $S$ appears.
\item $S$ is the composition of two subsystems, $\overline{P}$ and $E$, where $\overline{P}$ is the system we actually wish to study, while $E$ consists of all the experimental equipment - probes, detectors, and other measuring devices - used to observe $\overline{P}$. Clearly, $\overline{P}$ and $E$ interact with one another; so neither $\overline{P}$, nor $E$ are isolated systems; while $S=\overline{P} \vee E$ $is$ isolated if $E$ is chosen appropriately large; see e.g. \cite{FaFS}. 
\end{enumerate}

We propose to analyze a concrete model system where $\overline{P}$ is a quantum dot in a semi-conductor device containing a component, $P$, close to $E$ that can  bind up to $N < \infty$ electrons. A sketch of this system is given in Figure $1$.
\begin{figure}[H]
\begin{center}
\includegraphics[scale=0.42]{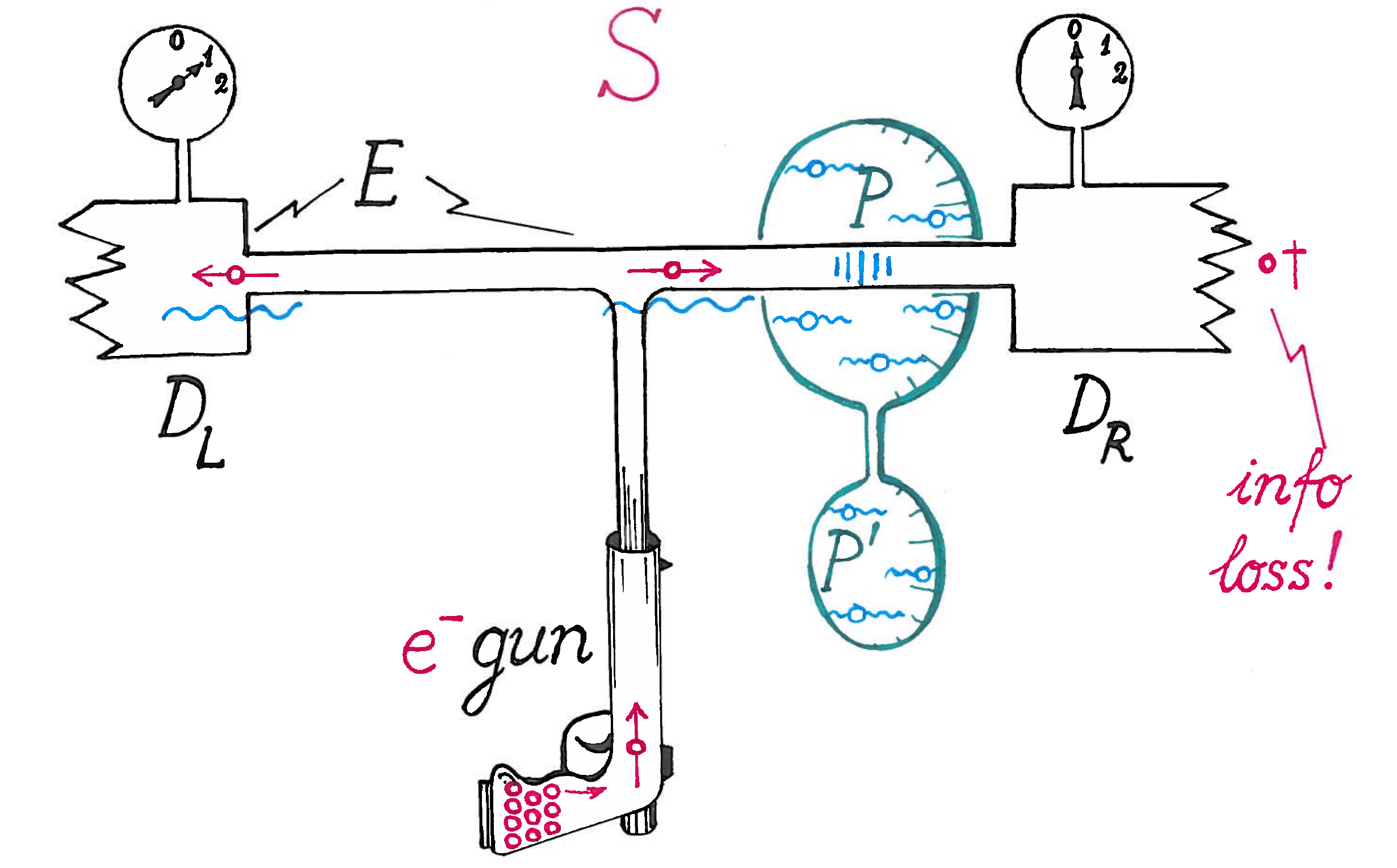}
\caption{The experimental setup.}
\end{center}
\end{figure}
 Electrons can enter into the component $P$ or tunnel out of it and into a component $P'$, (i.e., $P$ can be ionized). However, processes changing the number, $\nu$, of electrons bound by $P$ are very slow (in a sense to be specified below). The experimental equipment  $E$  consists of a conducting channel close to the region where $P$ is located, with electron detectors, $D_L$ and $D_R$, at both ends of the channel, and of an electron gun that shoots electrons into the channel at regular time intervals; (e.g., one electron every $\tau$ seconds, where $\tau$ is such that, typically,  there is only one  electron in the channel). The electrons in the channel connecting $D_L$ and $D_R$ serve as probes to  observe $P$ indirectly: the rates at which electrons in the channel hit $D_L$ or $D_R$, respectively, depend on the number, $\nu$, of electrons  bound by $P$. These latter electrons  give rise to a ``Coulomb blockade'' in the branch of the channel leading towards $D_R$: the larger $\nu$, the more likely it is that an electron propagating through the channel will hit the detector $D_L$.

The \textit{only direct observation} of  $S$ that can be made is to see whether an electron in the channel triggers a click of either $D_L$ or $D_R$. In an idealized (simplified) mathematical description of $S$, this observation will be represented by an operator, $X$, with 
\begin{equation}
X= \mathds{1}_{\overline{P}} \otimes \left( \begin{array}{cc} \mathds{1}&0\\0&-\mathds{1} \end{array} \right)_{E} ,
\end{equation}
with the property that $X$ has the eigenvalue $1$ on all channel states describing only $one$ electron propagating through the channel that will hit $D_L$ and the eigenvalue $-1$ on all channel states describing $one$ electron hitting $D_R$. 
Failed measurements will not be recorded. When an electron is scattered onto a detector it disappears for ever, and the original state of the detector is restored after a time much shorter than the cycle time $\tau$ of the electron gun. The physical quantity represented by the operator $X$ is measured around times $\tau,2 \tau, 3\tau,...$. These measurements are \textit{projective measurements}, as described in \cite{FS}. 
If the electron gun and the detectors function perfectly, only the eigenvalues $\xi=\pm 1$ are observed in every measurement. We thus denote by $\xi_j= \pm 1$ the outcome of the $j^{th}$ successful measurement of $X$; (with $\xi_j=1$ if $D_L$ clicks, and $\xi_j=-1$ if $D_R$ clicks). As is evident from Figure 1, the probability of the outcome $\xi_j=1$ in the $j^{th}$ measurement of $X$ depends on the number $\nu$ of electrons bound by the component $P$ of the quantum dot $\overline{P}$ close to the conducting channel. Because of the effect of the Coulomb blockade on the motion of electrons in the channel, the probability of the outcome $\xi_j=1$ increases with increasing $\nu$. It is convenient to introduce an operator 
\begin{equation}
\mathcal{N}:=\mathcal{N}_P \otimes  \mathds{1}_{P'} \otimes \mathds{1}_E
 \end{equation}
 whose eigenvalues $\nu \in \lbrace0,1,2,...,N\rbrace$ correspond to the number of electrons bound by the component $P$. Clearly, at a fixed time, $\mathcal{N}$ commutes with $X$. We assume that the Heisenberg time evolution of $\mathcal{N}$ is very slow, in the sense that the norm of
 \begin{equation}
 \Vert \frac{d}{dt}\mathcal{N}(t)\Vert
 \end{equation}
 is tiny for all $t$, where $\mathcal{N}(t)$ is the Heisenberg-picture operator corresponding to $\mathcal{N}$ at time $t$. We begin by studying the special case where $\mathcal{N}$ $commutes$ with the time evolution, i.e., $\mathcal{N}(t)$ is independent of $t$. One then speaks of a ``non-demolition measurement'' of the number of electrons in $P$. 
 
 For our model, a measurement protocol of length $k$ consists of a sequence 
 \begin{equation}
\underline{\xi}^{(k)} = (\xi_1,\xi_2,...,\xi_k),
 \end{equation}
 $k=1,2,3,....,$ of outcomes, $\xi_j = \pm 1$, of $k$ successful successive projective measurements of $X$; see, e.g., \cite{Griffiths}. (Recall that  unsuccessful measurements of $X$ are not recorded, and we will henceforth ignore them.) Quantum mechanics enables us to calculate the probability, denoted by $\mu_{\omega}(\xi_1,\xi_2,...,\xi_k)$, of a measurement protocol $\underline{\xi}^{(k)}$, given the state $\omega$ of the system $S$ just before measurements of $X$ have started.
 
 The question we wish to address is what a measurement protocol $\underline{\xi}^{(k)}$ tells us about the number, $\nu$, of electrons bound by $P$. The answer is actually fairly obvious. Given $\nu$, one calculates the repulsive potential acting on an electron in the right branch of the channel near the region where $P$ is located. Simple quantum-mechanical calculations then yield the \textit{Born probabilities}, $p(L \vert \nu)$ and $p(R \vert \nu)$, for an electron in the channel to be scattered onto $D_L$ or onto $D_R$, respectively. Clearly,
 \begin{equation}\label{4}
 p(L \vert \nu)+p(R \vert \nu)=1,
 \end{equation}
 and if the detectors function perfectly  $either$ $D_L$ $or$ $D_R$ will click, each time the electron gun has shot an electron into the channel. Given a measurement protocol $\underline{\xi}^{(k)}$, let 
 \begin{equation}\label{5}
 f^{(l,k)}_{L}(\underline{\xi}):= \frac{1}{k-l} \sharp \big \{ j \in \{l+1,...,k\} \mid \xi_j=1  \big\}
 \end{equation}
 be the frequency that an electron -- out of $k-l$ electrons shot into the channel after the $l^{th}$ shot of the gun  -- is scattered onto $D_L$, and let $f^{(l,k)}_R(\underline{\xi})$ be the frequency that an electron is scattered onto $D_R$, with $\underline{\xi}:=(\xi_1, \xi_2,...) = (\xi_i)_{i \in \mathbb{N}}$. Then 
 \begin{equation}
 f^{(l,k)}_{L}(\underline{\xi}) + f^{(l,k)}_{R}(\underline{\xi})=1.
 \end{equation}
 The law of large numbers suggests that the limit 
 \begin{equation} \label{7}
 f^{(\infty)}_{L/R}(\underline{\xi}):= \underset{k \rightarrow \infty }{\lim} f^{(l,k)}_{L/R}(\underline{\xi})
 \end{equation}
exists for any choice of $l$, $\mu_{\omega}$- a.s.. The value of  $f^{(\infty)}_{L/R}(\underline{\xi})$ is independent of $l$ and of the first $m$ measurement outcomes, $\xi_1,...,\xi_m$, for $any$ $m < \infty$. Such a function is called an ``observable at infinity'', or a ``pointer observable''. Now, if the state of $\overline{P}$ happens to be an eigenstate of the operator $\mathcal{N}$ corresponding to the eigenvalue $\nu$ (i.e., $\nu$ electrons are bound by $P$, and tunneling between $P$ and $P' = \overline{P}\setminus P$ is suppressed) then one expects that
\begin{equation}
f^{(\infty)}_{L/R}(\underline{\xi})=p (L/R \vert \nu), 
 \end{equation}
 as quantum theory predicts.
 Once we understand how fast the limit $k \rightarrow \infty$ on the right side of \eqref{7} is approached, and assuming that $p(L \vert \nu)$ separates points of the set $\{1,...,N\}$, meaning that 
 \begin{equation}\label{9}
 \underset{\nu_1 \neq \nu_2}{\min} \vert p(L \vert \nu_1) - p(L \vert \nu_2) \vert =: \kappa >0,
 \end{equation}
 then we will have found the answer to the question what a measurement protocol $\underline{\xi}^{(k)}$ is likely to tell us about the number $\nu$ of electrons bound by $P$, assuming that $k$ is large enough: 
 $\nu$ can be inferred from $f^{(l,k)}_{L}(\underline{\xi})$, (with $l$ fixed), by using \eqref{7}-\eqref{9}, with an error rate that tends to $0$, as $k \rightarrow \infty$!
 
 Actually, things are a little more complicated, but only slightly. If $\overline{P}$ consists of two or more spatially separated components joined by tunneling junctions, with only $P$ close to the channel connecting $D_L$ and $D_R$, then the quantum state of $\overline{P}$ is usually \textit{not} an eigenstate of the operator $\mathcal{N}$ even if that state is pure.  The probability for $\mathcal{N}$ to have the value $\nu \in \{1,2,..., N \}$ -- when measured -- is determined by the quantum state of $\overline{P}$ according to Born's rule.  A rather straightforward argument shows that, after the passage of a large number of electrons through the channel joining $D_L$ and $D_R$, the state of $\overline{P}$ approaches an \textit{incoherent mixture} of eigenstates of $\mathcal{N}$, i.e., of states with a definite number of electrons bound by $P$.  This phenomenon is called \textit{``decoherence''}. Once decoherence has set in, the claim after \eqref{9} becomes valid. 
 
 Our assumption that $\mathcal{N}$ commutes with time evolution, i.e., that a \textit{non-demolition measurement} of $\mathcal{N}$ is carried out, is only justified in the limit where the probability for an electron to tunnel into or out of $P$  tends to $0$. While this tunneling probability may be very small, it is usually non-zero. We should therefore ask what can be said about the time evolution of the eigenvalue $\nu(t)$ of the operator 
 $\mathcal{N}(t)$ (with $t=time$) inferred from long but finite measurement protocols recording the outcome of repeated direct measurements of the observable represented by the operator $X$.  We will show that  the average number of electrons bound by $P$ during time intervals of length $T\gg \tau$, as revealed by very long protocols of measurements of $X$, can be described by some stochastic  jump process, 
 $(\nu(t))_{t=jT,j=0,1,2,...}$, on the spectrum $\{1,...,N \}$ of the operator $\mathcal{N}$, where $T$ is determined by the speed of convergence in \eqref{7} and the amount by which $\mathcal{N}(t)$ varies over the cycle time $\tau$ of the electron gun, with $T \rightarrow \infty$ in  the limit where $\mathcal{N}(t)$ is independent of $t$ (non-demolition measurement of $\mathcal{N}$).  A fairly elementary result on the process $(\nu(j T))_{j=0,1,2,...}$ is presented in  Section \ref{Section4}. More refined results on the properties of the jump process $(\nu(jT))_{j=0,1,2,...}$ will appear in forthcoming work.

We pause to comment on the analogy between the quantum mechanics of sequences of projective measurements, as described above, and the statistical mechanics of classical lattice gases: A measurement protocol $\underline{\xi}^{(k)}$ corresponds to a configuration of particles of the lattice gas; ($\xi_{j}=1 \leftrightarrow$ a particle occupies site $j\in \mathbb{Z}_{+}$, $\xi_{j}=-1\leftrightarrow$ site $j$ is empty). The quantum-mechanical probability, $\mu_{\omega}(\xi_1,\xi_2,...,\xi_k)$, of observing the sequence 
$\underline{\xi}^{(k)} = (\xi_1,...,\xi_k)$ of measurement outcomes corresponds to the probability of the corresponding configuration of particles, as predicted by the Gibbs equilibrium 
measure of the lattice gas. The frequency  $f^{(0,k)}_{L}(\underline{\xi})$ corresponds to the density of particles in the subset $\lbrace 1,2,...,k \rbrace$ of the one-dimensional lattice $\mathbb{Z}_{+}$, and $f^{(\infty)}_{L}(\underline{\xi})$ corresponds to the infinite-volume density of particles. The number $\nu$ of electrons bound by $P$ corresponds to the chemical potential of particles in the lattice gas. The equivalence between the canonical and the grand-canonical ensemble in the thermodynamic limit of the lattice gas corresponds to the statement that 
$f^{(\infty)}_{L}(\underline{\xi})$ has a sharp value, $p(L \vert \nu)$, and that the fluctuations of 
$f^{(0,k)}_{L}(\underline{\xi})$ are of order $O(k^{-\frac{1}{2}})$, as $k\rightarrow \infty$; etc. 
 
 We conclude this introduction by adding a little more precision to the statements sketched above. Let $\Xi$ denote the space of infinitely long measurement protocols (equipped with the $\sigma$-algebra generated by the cylinder sets
 \begin{equation}\label{10}
 \Sigma(\eta_{j_1},..., \eta_{j_m}):= \{ \underline{\xi} \in \Xi \mid \xi_{j_i}= \eta_{j_i}, i=1,...,m \},
 \end{equation} 
 for arbitrary choices of  $(\eta_{j_1},..., \eta_{j_m})$, $m< \infty$). Given $\underline{\xi} \in \Xi$, we denote by $\underline{\xi}^{(l,k)}$ the sequence of measurement outcomes ($\xi_{l+1},\xi_{l+2},..., \xi_k)$, $l<k$, and, apparently, $\underline{\xi}^{(k)}=\underline{\xi}^{(0,k)}$. As has been indicated above, the choice of a state, $\omega$, of the quantum system $S= \overline{P} \vee E$ equips the measure space $\Xi$ with  a probability measure, $\mu_{\omega}$, defined on the $\sigma$-algebra generated by the cylinder sets $\Sigma(\eta_{j_1},...,\eta_{j_m})$. It associates with a protocol $\underline{\xi}^{(k)}$  the generalized Born probability 
 \begin{equation}\label{11}
 \mu_{\omega}(\xi_1,...,\xi_k)
 \end{equation}
 for the first $k$ projective measurements of the observable $X$ (carried out at times $\approx j\tau $, $j=1,...,k$) to yield the values $\xi_1,...,\xi_k$. The unique expression for $\mu_{\omega}( \xi_1,...,\xi_k)$ provided by quantum mechanics was first found by Schwinger in \cite{Schw} and rediscovered by Wigner \cite{Wig} and, most certainly, by very many other theorists. It is recalled in Sect.  \ref{Sec21}, Eq. \eqref{LSW}. For consistency of $\mu_{\omega}$, it is necessary that 
 \begin{equation}
 \sum_{\xi_{k+1}} \mu_{\omega}(\xi_1,...,\xi_k,\xi_{k+1})= \mu_{\omega}(\xi_1,...,\xi_k),
 \end{equation}
 with $\sum_{\xi} \mu_{\omega}(\xi)=1$, so that $\mu_{\omega}$ is a probability measure on $\Xi$. These conditions hold always true in quantum mechanics. One then observes that 
 \begin{equation}\label{13}
 \mu_{\omega}(\Sigma(\eta_{j_1},..., \eta_{j_m}))= \underset{\underset{\xi_{j_i}= \eta_{j_i},i=1,...,m}{\xi_1,...,\xi_{j_m}}}{\sum} \mu_{\omega}(\xi_1,...,\xi_{j_m}).
 \end{equation}
 
 The general theory of direct measurements/observations in quantum mechanics requires that $\mu_{\omega}$ also satisfies a certain ``decoherence condition'', which says, roughly speaking, that subsequent measurements of the observable $X$ are ``independent'' of one another, (which does not mean that the measurements outcomes $\xi_1,...,\xi_k$,... are uncorrelated). For the purposes of this introduction we do not need to know precisely what is meant by decoherence and why it is important; (but see Sect. \ref{Section3} and, e.g., \cite{Griffiths, FS}). A consequence of decoherence is that 
 \begin{equation}\label{14}
 \sum_{\xi_1,...,\xi_k} \mu_{\omega}(\xi_1,...,\xi_k, \xi_{k+1},....,\xi_{k+r})=\mu_{\omega}(\xi_{k+1},...,\xi_{k+r}),
 \end{equation}
 (possibly only up to a tiny error), for arbitrary $k, r \in \mathbb{N} $, where the right side of \eqref{14} is the generalized Born probability for the outcomes $\xi_{k+1},..., \xi_{k+ r}$ in measurements $k+1,...,k+r$, assuming that measurements $1,...,k$ have not been carried out; (there have been no gun shots at times $\tau,2 \tau,...,k\tau)$. 
  
 Next, we define certain subsets of $\Xi$:
 \begin{equation}\label{15}
 \Xi_{\nu}^{(l,k)}( \underline{\varepsilon}):= \{ \underline{\xi} \in \Xi \mid \vert f_{L}^{(l,k)}(\underline{\xi})- p(L \vert \nu) \vert < \varepsilon_{k-l} \}, \qquad l<k,
 \end{equation}
 where, we recall, $ f_{L}^{(l,k)}(\underline{\xi})$ is the frequency of clicks of $D_L$ in measurements 
 $l+1,...,k$, (see \eqref{5}), and $p(L \vert \nu)$ has been defined right above \eqref{4}. In \eqref{15}, 
 $\underline{\varepsilon}=(\varepsilon_m)_{m=1}^{\infty}$ is a suitably chosen sequence of positive numbers converging to $0$, with $\sqrt{m}\varepsilon_{m} \rightarrow \infty$, as 
 $m \rightarrow \infty$. An appropriate choice of $\underline{\varepsilon}$ depends on details of the model used to describe the system $S$; (see Sects. \ref{Section2}, \ref{Section4} and the appendix for simple examples). We set
 \begin{equation} \label{16}
\Xi^{(l,k)}( \underline{\varepsilon})= \bigcup_{\nu=1}^{N} \Xi_{\nu}^{(l,k)}( \underline{\varepsilon}).
 \end{equation}
 The following claim is obvious.
 \vspace{1mm}
 
 \begin{enumerate}[(A)]
 \item If $k-l$ is so large that $\varepsilon_{k-l}< \kappa/2$, where $\kappa= \underset{\nu_1 \neq \nu_2}{\min} \vert p(L \vert \nu_1)- p(L \vert \nu_2) \vert$, see \eqref{9}, then the sets  $\Xi_{1}^{(l,k)}( \underline{\varepsilon}),...,\Xi_{N}^{(l,k)}( \underline{\varepsilon})$ are all disjoint from one another.
 \end{enumerate}
 \vspace{1mm}
 
 Let $\mu_{\omega} (\xi_{l+1},...,\xi_{k} \vert \eta_{j_1},...,\eta_{j_m})$, $(j_m \leq l)$, denote the conditional probability of measurement outcomes $\xi_{l+1},..., \xi_{k}$, given that the outcomes in measurements $j_1<j_2<...<j_m \leq l$ were $\eta_{j_1}$,...,$\eta_{j_m}$, respectively. ( As an aside, we note that 
 \begin{equation}\label{17}
 \mu_{\omega}(\xi_{l+1},..., \xi_{k} \vert \eta_{j_1},...,\eta_{j_m})=\frac{\mu_{\omega}(\Sigma(\eta_{j_1},...,\eta_{j_m},\xi_{l+1},...,\xi_{k}))}{\mu_{\omega}(\Sigma(\eta_{j_1},...,\eta_{j_m}))},
 \end{equation}
 with $\Sigma(...)$ as in \eqref{10}). Under quite general hypotheses on the model used to describe $S$, we will show:
 \begin{enumerate}[(B)]
 \item The conditional probabilities of the complement of the set $\Xi^{(l,k)}( \underline{\varepsilon})$ satisfy a bound 
 \begin{equation} \label{18}
  \mu_{\omega}( (\Xi^{(l,k)}( \underline{\varepsilon}))^c  \vert \eta_{j_1},...,\eta_{j_m}) =1-   \mu_{\omega}( \Xi^{(l,k)}( \underline{\varepsilon})  \vert \eta_{j_1},...,\eta_{j_m})  < \delta_{k-l},
 \end{equation}
 for a sequence $\underline{\delta}=(\delta_n)_{n=1}^{\infty}$ converging to zero, as $n \rightarrow \infty$, uniformly in $\eta_{j_1},...,\eta_{j_m}$, for each fixed $l$.
 \end{enumerate}
 \vspace{1mm}
 
 The bound \eqref{18} means that the frequencies,  $f_{L}^{(l,k)}(\underline{\xi})$, of long, but finite measurement protocols $\xi_{l+1},..., \xi_{k}$, are almost always very close to one of the numbers $p(L \vert 1),..., p(L \vert N)$, independently of what has been measured in earlier measurements. 
 
 Statements (A) and (B) are analogous to the statement that, in the statistical mechanics of classical lattice gases, the canonical and grand canonical ensembles are equivalent in the thermodynamic limit, independently of boundary conditions, and that the limiting Gibbs states are ``mutually singular'', as some thermodynamic parameters, such as the chemical potential of particles in a lattice gas -- the analogue of 
 $\nu \in \{1,...,N \}$ -- is varied. The mathematics used to prove these statements is similar to the one used to prove (B) and is borrowed from methods used in ``hypothesis testing'', (see \cite{Chernoff}, \cite{JaPi12}).
  
We say that a measurable function $h$ is an ``observable at $\infty$'' (or a ``pointer observable'') for $S$ iff $h$ is well defined and bounded except on sets of $\mu_{\omega}$-measure $0$, and the value, $h(\underline{\xi})$, of $h$ at an arbitrary point $\underline{\xi} \in \Xi$ is independent of the first $m$ components, $\underline{\xi}^{(m)}$, of $\underline{\xi}$, for any $m<\infty$, almost surely with respect to the measure $\mu_{\omega}$, where 
$\omega$ is an arbitrary state of $S$. An example of an observable at $\infty$ is the limiting frequency 
 $f_{L}^{(\infty)}(\underline{\xi})$ defined in \eqref{7}. Observables at infinity span a commutative algebra of functions on $\Xi$, denoted by $\mathcal{O}_{\infty}[S]$. Let us assume, for simplicity, that $\mathcal{O}_{\infty}[S]$ is finite-dimensional. Then it is the algebra of functions on a finite set 
 \begin{equation} \label{19}
 \Xi_{\infty} \simeq \{1,...,N_{\infty} \},
 \end{equation} 
for some finite integer $N_{\infty}$. (More generally, $\mathcal{O}_{\infty}[S]$ is the algebra of bounded measurable functions on a compact Hausdorff space, $\Xi_{\infty}$; see e.g. \cite{FS}. But \eqref{19} holds in the model described above, with $N_{\infty}=N$; see also Section 3.) It then follows that, for an arbitrary state $\omega$ of $S$, the measure $\mu_{\omega}$ on 
 $\Xi$ associated to $\omega$ has the unique decomposition
 \begin{equation}\label{20}
 \mu_{\omega}(\cdot)= \sum_{\nu \in \Xi_{\infty}} P_{\omega}(\nu) \mu_{\omega}(\cdot\vert \nu),
 \end{equation}
 where the measures $\mu_{\omega}(\cdot \vert \nu_1)$ and $\mu_{\omega}(\cdot \vert \nu_2)$ are mutually singular, for distinct values of $\nu_1$ and $\nu_2$, and $P_{\omega}(\nu) \geq 0$, for all $\nu$, with 
 $\sum_{\nu} P_{\omega}(\nu)=1$.
 
 The point of these remarks is that, even without knowing, a priori, which physical quantity, $A$, of $S$ can be measured indirectly, in a non-demolition measurement, with the help of very many successive direct measurements of the quantity represented by $X$, the commutative algebra $\mathcal{O}_{\infty}[S]$ of functions measurable at infinity tells us something about $any$ such quantity: The operator $A$ representing it corresponds to a continuous function on $\Xi_{\infty}$. (In the example of the system discussed above, $A=\mathcal{N}$, where $\mathcal{N}$ is the electron number operator associated with $P$.) 
 A point $\nu \in   \Xi_{\infty}$ is henceforth called a \textit{fact} (concerning the system $S$). 
 It can be reconstructed from very long measurement protocols. Results on ``hypothesis testing'' are 
 useful to understand how well functions measurable at infinity  can be approximated by sequences of functions of \textit{finite} measurement protocols, and how rapidly the latter converge to the former, 
 as the length of the measurement protocols tends to $\infty$. (A detailed study of these matters will 
appear in forthcoming work; but see Sects. 2 and 3.) 
 
 We now return to analyzing the system $S$ introduced above; see Figure 1 and Eqs. (1.1) - (1.5). We consider the situation where the operator $\mathcal{N}(t)$ slowly evolves in time $t$. We consider subsets, $\Xi_{\nu_1,...,\nu_p}(k;  \underline{\varepsilon})$, with the property that the frequencies, 
 $f_{L}^{(ik-k,ik)}(\underline{\xi})$, of clicks of the detector $D_L$ in measurements $ik-k+1,..., ik$ is within 
 $\varepsilon_{k}$ from the Born probability $p(L \vert \nu_i)$ introduced above \eqref{4}, for all $i=1,...,p$. If the number operator $\mathcal{N}$ \textit{commutes} with time evolution then the $\mu_{\omega}$-measure of the complement of the set 
 \begin{equation} \label{21}
 \Xi_{p}(k; \underline{\varepsilon}):= \bigcup_{\nu_1,...,\nu_p} \Xi_{\nu_1,...,\nu_p}( k; \underline{\varepsilon})
 \end{equation}
 is tiny ($< p\delta_{k}$, where $(\delta_n)_{n=1}^{\infty}$ is as in \eqref{18}), for an arbitrary state $\omega$ of 
 $S$. If $\mathcal{N}(t)$ evolves in time $t$ very slowly, e.g. if 
 \begin{equation*}
 \underset{0 \leq s \leq \tau}{\max} \| \mathcal{N}(t+s)- \mathcal{N}(t) \| \ll 1,
 \end{equation*}
 then the $\mu_{\omega}$-measure of the complement of $ \Xi_{p}(k; \underline{\varepsilon})$ remains very small:
 \begin{equation}\label{22}
 \mu_{\omega}( \Xi_{p}(k; \underline{\varepsilon})^c) = 1- \mu_{\omega}( \Xi_{p}(k; \underline{\varepsilon}))<2 p \delta_k \ll 1,
 \end{equation}
 provided $k$ is large, but $p$ is not too large. This fact enables us to introduce a measure, $\mathcal{P}_{\omega}$, on finite ``quantum trajectories'' (or ``histories of facts''), $\{ \nu(t_i) =: \nu_{i }\mid t_i=ik \tau\rbrace_{i=1}^{p}$, with $p$ such that $2 p \delta_{k} \ll 1$:
 \begin{equation}\label{23}
 \mathcal{P}_{\omega}(\nu_1,...,\nu_p):= 
 \mu_{\omega}(\Xi_{\nu_1,...,\nu_p}( k; \underline{\varepsilon})).
 \end{equation}
 By \eqref{22},
 \begin{equation} \label{24}
 \sum_{\underset{i=1,...,p}{\nu_{i}=1,...,N}} \mathcal{P}_{\omega}(\nu_1,...,\nu_p) \geq 1-2p \delta_{k}
  \approx 1.
 \end{equation}
 It is interesting to characterize the measures $\mathcal{P}_{\omega}$ more precisely in certain limiting regimes. 
 
 Section 2 contains an outline of general concepts and new results concerning indirect and, in particular, non-demolition measurements. 
 In Section 3, we introduce a family of concrete models describing systems $S$ of the type discussed above; (see \cite{BaBe,BaBeBe2,BaBeTi}). For these models, the general picture sketched above can be translated into rigorous mathematics. 
 In Section 4, we study indirect measurements of quantities evolving slowly in time, and we analyze trajectories of ``quantum jumps'' (described, in a suitable limiting regime, by Markov jump processes).
 
To conclude this introduction, we remark that, in many quantum systems tracked with the help of direct observations of probes, these probes are coherent pulses of field quanta of a wave medium, such as the quantized electromagnetic field, rather than electrons in a conducting channel. An example is the indirect observation of the trajectory of a charged particle ($\overline{P}$) by means of light scattering (E). 
A mathematically precise analysis of such systems is, however, quite involved and is postponed to a future publication.\\

{\bf{Acknowledgements.}} We thank Ph. Blanchard for numerous helpful discussions and encouragement, and M. Bauer and D. Bernard for some correspondence at an early stage of this work and some discussions. M. B. thanks  Jos\'{e} Luis \'{A}ngel P\'{e}rez Garmendia for  helpful discussions.  J.F. thanks the Institut des Hautes Etudes Scientifiques
for hospitality during final stages of work on this paper. The research of B.S. is supported in parts by the r\'{e}gion Lorraine.

 \section{Indirect retrieval of information on quantum systems\\
 through sequences of projective measurements} \label{Section2}
 In this section we explain in general terms how interesting information on a quantum system $S$ can be retrieved from long sequences of projective measurements of a possibly very uninteresting looking property of $S$. Our analysis can be viewed as a contribution to a general theory of ``non-demolition measurements/observations''. In this section, it is carried out within the standard  Hilbert-space formalism of quantum mechanics; (although using a more abstract algebraic formulation of quantum mechanics involving $C^*$-algebras of operators would actually be more natural -- see, e.g., \cite{FS}).  We consider an isolated quantum system $S$  characterized by the  data $(\mathcal{H}_S, \{U(t,s)\}_{t,s \in \mathbb{R}}, \mathcal{O}_S)$,  where 
\begin{itemize}
\item  $\mathcal{H}_{S}$ is the Hilbert space of pure state vectors of $S$.
\item  $\{U(t,s)\}_{t,s \in \mathbb{R}}$ is a family of unitary operators  on $\mathcal{H}_S$ representing time evolution. They satisfy the composition rule
\begin{equation}
U(t,s)=U(t,r)U(r,s), \qquad \forall \text{   } r, s, t \in \mathbb{R}.
\end{equation}
\item $\mathcal{O}_S \subset \mathcal{B}(\mathcal{H}_S)$ is a set of self-adjoint operators on 
$\mathcal{H}_S$ representing physical quantities of $S$ that can be measured/observed  in projective measurements. This set has the property that if an operator $X$ belongs to $\mathcal{O}_S$ and $f$ is a real-valued, bounded, continuous function on $\mathbb{R}$ then $f(X)$ belongs to $\mathcal{O}_S$, too.
 \end{itemize}
In many examples of concrete physical systems $S$, the set $\mathcal{O}_S$ is generated by only a few (often finitely many) operators. To simplify matters, we will assume that $\mathcal{O}_S$ is $abelian$, i.e. that it is generated by commuting, self-adjoint,  bounded operators. In this case, we may identify 
$\mathcal{O}_S$ with the \textit{commutative algebra} generated by all operators in 
 $\mathcal{O}_S$. In this paper we will usually assume that the spectrum, $\sigma_S$, of the commutative algebra 
 $\mathcal{O}_S$ is a \textit{countable} set of points. We then denote by $\pi_{\xi}$ (the orthogonal projection on $\mathcal{H}_S$ corresponding to) the characteristic function of the set $\{\xi\}$, where $\xi$ is an arbitrary point in the spectrum $\sigma_S$ of $\mathcal{O}_S$. The operators  $(\pi_{\xi})_{\xi \in \sigma_S}$ satisfy
 \begin{equation}\label{2.2}
 \sum_{\xi \in \sigma_S} \pi_{\xi} = \mathds{1}_{\mathcal{H}_S}, \qquad \pi_{\xi} \pi_{\xi'}= \delta_{\xi,\xi'} \pi_{\xi}. 
 \end{equation}
All operators in $\mathcal{O}_S$ can be written as linear combinations of the projections $(\pi_{\xi})_{\xi \in \sigma_S}$.
 
 \subsection{Calculus of frequencies and observables at infinity} \label{Sec21}
 We suppose that successive projective measurements of quantities in $\mathcal{O}_S$ are carried out at  times $t_i$, with $t_i < t_{i+1}$, $i=1,2,3, ...$. Given a state $\omega(\cdot)= \mathrm{Tr}(\rho_{\omega} (\cdot))$  on $\mathcal{B}(\mathcal{H}_S)$, where $\rho_{\omega}$ is a density matrix specified at some time $t_0$, the LSW formula (for L\"{u}ders, Schwinger, Wigner; see \cite{Lue}, \cite{Schw}, \cite{Wig}, and \cite{Griffiths}) yields the a-priori probability,  
 $\mu_{\omega}(\xi_1,...,\xi_n)$, for  observing a sequence of measurement results, $\xi_i \in \sigma_S$, observed at times $t_i$, $i=1,2,...,n$, corresponding to ``events'' $(\pi_{\xi_1}(t_1),...,\pi_{\xi_n}(t_n))$, where 
 \begin{equation}
 \pi_{\xi}(t):= U(t_0,t)\pi_{\xi} U(t,t_0).  
 \end{equation}
 Indeed, the LSW formula tells us that
 \begin{equation} \label{LSW}
 \mu_{\omega}(\xi_1,...,\xi_k) = \mathrm{Tr}(\pi_{\xi_k}(t_k) \cdot \cdot \cdot \pi_{\xi_1}(t_1)\rho_{\omega}  \pi_{\xi_1}(t_1)\cdot \cdot \cdot\pi_{\xi_k}(t_k)).
 \end{equation}
\textit{Remark.} It is well known that the LSW formula should only be used if a suitable form of ``decoherence'' holds for histories of repeated projective measurements of quantities represented by operators in 
$\mathcal{O}_S$; see, e.g., \cite{Griffiths}, \cite{FS}, (where a theory of projective measurements has been developed), and references given there. We will not discuss these matters here, because they do not play an important role in the following.\\
It follows from \eqref{LSW} and from \eqref{2.2} that 
 \begin{equation}
 \sum_{\xi_k \in \sigma_S} \mu_{\omega}(\xi_1,...,\xi_k) = \mu_{\omega}(\xi_1,...,\xi_{k-1}), \qquad \text{for arbitrary  } k < \infty\\
\end{equation}
and
 \begin{equation}
  \sum_{\xi \in \sigma_S} \mu_{\omega} (\xi) = 1,
 \end{equation}
and hence the Kolmogorov consistency criterion implies that the functional $\mu_{\omega}$ extends to a probability measure on the  measure space $(\Xi, \Sigma)$, where $\Xi= \bigtimes_{i=1}^{\infty} \sigma_{S}^{(i)}$ is an infinite Cartesian product of copies of the spectrum $\sigma_S$ of 
$\mathcal{O}_S$ and  $\Sigma$ is the $\sigma$-algebra generated by the cylinder sets in $\Xi$. 
The $\sigma$-algebra $\Sigma$  is the Borel $\sigma$-algebra for the compact topological space $\Xi$.  
Let $L^{\infty}(\Xi,\mu_{\omega})$ be the usual space of equivalence classes of  bounded measurable functions with respect to the measure  $\mu_{\omega}$. This space contains a closed subspace, denoted by 
$\mathcal{O}_{\infty}[\omega]$, of bounded measurable functions, $h$, with the property that
$$h(\underline{\xi})=h(\underline{\xi}'), \text{   } \text{almost surely with respect to }    \mu_{\omega},$$
whenever  $\xi_i =\xi'_i$, except for finitely many $i$. Of course, $L^{\infty}(\Xi,\mu_{\omega})$ and 
$\mathcal{O}_{\infty}[\omega]$ are algebras; $\mathcal{O}_{\infty}[\omega]$ is called the ``algebra of observables at infinity'' (or ``algebra of pointer observables'').

Given a function $f \in L^{\infty}(\Xi,\mu_{\omega})$, we define a bounded linear functional on the algebra of observables at infinity by setting
    \begin{equation} \label{2.7}
  \mathbb{E}_{\omega}(f \vert h):= \int_{\Xi} f(\underline{\xi}) h(\underline{\xi})  d\mu_{\omega}(\underline{\xi}), 
    \end{equation}
 for any $h \in \mathcal{O}_{\infty}[\omega]$. If $f \geq 0$, then $   \mathbb{E}_{\omega}(f \vert \cdot )$ is actually a bounded \textit{positive} linear functional on the algebra $\mathcal{O}_{\infty}[\omega]$ and hence corresponds to a  regular  Borel measure  $\mu_{\omega}(f \vert \cdot )$ on the spectrum, denoted by 
 $\Xi_{\infty}$, of $\mathcal{O}_{\infty}[\omega]$. 
 By \eqref{2.7},  $\mathbb{E}_{\omega}(f \vert \chi_{\Delta} )=0$  if   
 $   \mathbb{E}_{\omega}(f\equiv 1 \vert \chi_{\Delta} )=0$, where $\Delta$ is  an arbitrary Borel  subset of 
 $\Xi_{\infty}$. Thus, the measure  $\mu_{\omega}(f \vert \cdot )$ is absolutely continuous with respect to  $\mu_{\omega}(1 \vert \cdot )=: P_{\omega}$. Consequently, the Radon-Nikodym derivative  of $\mu_{\omega}(f \vert \cdot )$ with respect to $\mu_{\omega}(1 \vert \cdot )$ is well-defined, and we have that 
   \begin{equation} \label{2.8}
    \mathbb{E}_{\omega}(f \vert 1) = \int_{\Xi_{\infty}} \frac{ d \mu_{\omega}(f \vert \cdot )}{d \mu_{\omega}(1 \vert \cdot )}(\nu) dP_{\omega}(\nu),
   \end{equation}
   for all positive functions $f \in L^{\infty}(\Xi,\mu_{\omega})$. This formula can  be extended to \textit{arbitrary} real functions, $f$, in $L^{\infty}(\Xi,\mu_{\omega})$ by writing $f=f_{+} - f_{-}$, with $f_{+}, f_{-} \geq 0$, and applying \eqref{2.8} separately to $f_{+}$ and $f_{-}$ and, subsequently to arbitrary complex-valued functions in $L^{\infty}(\Xi,\mu_{\omega})$.  Following the proof of \cite[Th. 5.7]{Rao}, one can find a set $N \subset \Xi_{\infty}$ of measure zero and select for each $f$ a representative of the Radon-Nikodym derivative  such that, on the complement of $N$,  
 $ \frac{ d \mu_{\omega}(f \vert \cdot )}{d \mu_{\omega}(1 \vert \cdot )}(\nu)$ is linear in $f$, positive whenever $f\geq 0$, and bounded by $\Vert f \Vert_{\infty}$. Hence we may rewrite Eq. \eqref{2.8} as
   \begin{equation} \label{2.9}
      \mathbb{E}_{\omega}(f \vert 1) = \int_{\Xi_{\infty}} \left( \int_{\Xi} f(\underline{\xi}) d\mu_{\omega}(\underline{\xi} \vert \nu)\right)  dP_{\omega}(\nu), 
   \end{equation}
    where $\mu_{\omega}( \cdot  \vert \nu)$ is a Borel probability measure on $\Xi$, for $P_{\omega}$- a.e. $\nu \in \Xi_{\infty}$.
 
 Before studying properties of  the ``extremal'' measures $\mu_{\omega}( \cdot  \vert \nu)$, 
 $\nu \in \Xi_{\infty}$, we describe an explicit example of an observable at infinity: We consider a bounded measurable function 
 $f: \sigma_{S}^{\times l} \rightarrow \mathbb{C}$ and define ``frequencies'' $f^{(k)}$, $k=1,2,3,...$,  on $\Xi$ by  
    \begin{equation*}
    f^{(k)}(\underline{\xi}):=\frac{1}{k} \sum_{j=1}^{k} f_{j}(\underline{\xi}), \qquad \text{where}
    \end{equation*} 
    \begin{equation} \label{2.10}
    f_{j}(\underline{\xi}):=f(\xi_j, ... , \xi_{j+l-1}).
    \end{equation}
If  the sequence of functions $(f^{(k)})_{k=1}^{\infty}$ converges $\mu_{\omega}$- a.e., as $k\rightarrow \infty$, then the limiting function, denoted by $f^{(\infty)}$, belongs to 
$\mathcal{O}_{\infty}[\omega]$. Sufficient conditions  for convergence may be derived from ergodicity hypotheses. Assuming, for example, that, for arbitrary $n<\infty$, the ``decoherence condition''
\begin{equation} \label{ergo}
\sum_{\xi_1,...,\xi_n \in \sigma_S} \mu_{\omega}(\xi_1,...,\xi_n, \underline{\xi}^{(> n)})=\mu_{\omega}(\underline{\xi}^{(> n)}),
\end{equation}
holds, the ergodic theorem implies that $f^{(k)}$ converges $\mu_{\omega}$- a.e., as $k\rightarrow \infty$, to an          element in
 $\mathcal{O}_{\infty}[\omega]$. For an autonomous system $S$ (i.e., for one with $U(t,s)=U(t-s)$, for all $t,s \in \mathbb{R}$) with the additional property that there is ``perfect decoherence'', for arbitrary sequences of projective measurements of quantities represented by operators in $\mathcal{O}_S$, Eq. \eqref{ergo} holds.

\subsection{Hypothesis testing and emergence of facts} \label{hypo}
Let  $f: \sigma_{S}^{\times l} \rightarrow \mathbb{C}$ be a function depending on finitely many measurement outcomes, and let $f^{(k)}$ denote the frequency associated with $f$, as defined in \eqref{2.10}. We decompose $\Xi_{\infty}$ into the level sets of the limiting function $f^{(\infty)}$:
\begin{equation} \label{2.12}
\Xi_{\infty} = \bigcup_{\alpha=1}^{N} \Xi_{\infty}^{\alpha},
\end{equation}
where $\Xi_{\infty}^{\alpha}$ is a level set of $f^{(\infty)}$. We henceforth assume that
\begin{equation} \label{2.13}
\text{min}_{\alpha_{1}\neq \alpha_{2}} \vert f^{(\infty)}|_{\Xi_{\infty}^{\alpha_1}} - f^{(\infty)}|_{\Xi_{\infty}^{\alpha_2}}\vert =: \kappa > 0.
\end{equation}
We define
\begin{equation} \label{2.14}
d\mu_{\omega}^{\alpha}:= \int_{\Xi_{\infty}^{\alpha}} d\mu_{\omega}(\cdot\vert \nu) 
dP_{\omega}(\nu).
\end{equation}
Note that the measures $d\mu_{\omega}^{\alpha}$ are mutually singular for different values of $\alpha \in \lbrace 1,...,N \rbrace$.
We introduce the ``fluctuation variables''
\begin{equation} \label{2.15}
\phi^{(k)}_{\alpha}:= \sqrt k \left( f^{(k)} - m_{\alpha} \right), \qquad
\text{where } \text{   } m_{\alpha} := f^{(\infty)}(\underline{\xi})\vert_{\Xi_{\infty}^{\alpha}}.
\end{equation}
Our analysis is based on the assumption that $(d\mu_{\omega}^{\alpha}, \phi^{(k)}_{\alpha})$ satisfies a suitable Central Limit Theorem \cite{esseen, Leb14,Feller}, for all $\alpha=1,2,...,N$. We consider the generating functions, $F_{\alpha}^{(k)}(h)$, of connected correlations of fluctuation variables defined by
\begin{equation} \label{2.16}
F_{\alpha}^{(k)}(h):= \text{ln}\left( \int_{\Xi} d\mu_{\omega}^{\alpha}(\underline{\xi})e^{h\phi_{\alpha}^{(k)} (\underline{\xi})} \right),
\end{equation}
where $h \in \mathbb{R}$. We assume that, for every  $\alpha=1,...,N$, there is an open interval 
$U \subset \mathbb{R}$ containing the origin such that
the sequence  $(F_{\alpha}^{(k)}(h))_{k=1}^{\infty}$ of real-valued functions, together with their first, second and third derivatives in $h$ converge to bounded, continuous functions of $h \in U$, as $k\rightarrow \infty$. Furthermore, we assume that -- to mention a concrete example of a suitable hypothesis --
\begin{itemize}
\item 
$\frac{d}{dh} F^{(k)}_{\alpha}(0) \rightarrow 0,$
\item
 $\frac{d^2}{dh^2} F^{(k)}_{\alpha}(0) \rightarrow \gamma_{\alpha}>0, $
and
\item 
$\frac{d^3}{dh^3} F^{(k)}_{\alpha}(h) \rightarrow 0, \text{ for all  } h \in U,$
\end{itemize}
as $k \rightarrow \infty$, for all $\alpha=1,2,...,N$. Then the fluctuation variables $\phi^{(k)}_{\alpha}$ approach \textit{Gaussian random variables}, $\phi^{(\infty)}_{\alpha}$, with variance $\gamma_{\alpha}$ centered at $0$. This is the contents of the Central Limit Theorem. \\
The hypotheses formulated above can be derived from suitable cluster properties of the cumulants 
$$ \mathbb{E}^{\alpha}_{h,k}(f_{j_1}; ... ;f_{j_m}), \qquad \text{for } m=2,3 \text{   }\text{and all } h\in U,$$
where $\mathbb{E}^{\alpha}_{h,k}$ denotes an expectation with respect to the probability measure
$$ \text{exp}[-F_{\alpha}^{(k)}(h)] \text{  } e^{h\phi_{\alpha}^{(k)}(\underline{\xi})}\text{  } 
d\mu^{\alpha}_{\omega}(\underline{\xi}).$$
To be somewhat more precise, these cumulants should be assumed to be summable in $j_2,...,j_m$, for every fixed $j_1$, with bounds that are uniform in $k$.

We introduce subsets 
\begin{equation} \label{2.16}
\Xi_{\alpha}^{(k)} (\underline{\varepsilon}):= \lbrace \underline{\xi} \vert \text{ }| f^{(k)}(\underline{\xi}) - 
m_{\alpha}\vert < \varepsilon_{k} \rbrace, 
\end{equation}
see Eq. \eqref{2.15}, where $(\varepsilon_k)_{k=1}^{\infty}$ is a sequence of positive numbers converging to $0$, with 
$\varepsilon_{k}\cdot \sqrt k \rightarrow \infty$, as $k \rightarrow \infty$.

The following theorem is an immediate consequence of \eqref{2.13} and of the Central Limit Theorem.

\begin{theorem} \label {2.2.1}
 1. If $k$ is so large that $\varepsilon_{k} < \frac{\kappa}{2}$ \text{   }then 
\begin{equation}
\Xi^{(k)}_{\alpha_1}( \underline{\varepsilon}) \cap \Xi_{\alpha_2}^{(k)}( \underline{\varepsilon}) = \emptyset, 
\end{equation}
for \text{   }  $\alpha_1 \neq \alpha_2$.

2. The measure of the complement of the set $\cup_{\alpha} \Xi_{\alpha}^{(k)}( \underline{\varepsilon})$ tends to $0$, as $k$ tends to $\infty$; i.e.,
\begin{equation}
\mu_{\omega}\left( (\cup_{\alpha}\Xi_{\alpha}^{(k)}( \underline{\varepsilon}))^{c}\right) \rightarrow 0,
\end{equation}
as $k \rightarrow \infty$, at a speed determined by the speed of convergence of the fluctuation variables 
$\phi^{(k)}_{\alpha}$ to Gaussian random variables.
\end{theorem} 

This theorem tells us that the value, $\alpha$, of $f^{(\infty)}$ can be inferred from long, but finite measurement protocols, $\underline{\xi}^{(k)}$, with an error that tends to $0$, as $k$ approaches $\infty$. The value of $f^{(\infty)}$ will henceforth be called a \textit{``fact''}. Apparently, long measurement protocols reveal \textit{facts} about the state of the system, and the chance of being mistaken about which fact is emerging tends to $0$, as the length $k$ of the measurement protocol tends to $\infty$.

The a-priori probability of observing an arbitrary measurement protocol $\underline{\xi}^{(k)}$ corresponding to the value $\alpha$ of $f^{(\infty)}$ approaches the value
\begin{equation} \label{2.20}
\mathcal{P}_{\omega}(\alpha):= \int_{\Xi^{\alpha}_{\infty}} dP_{\omega}(\nu)
\end{equation}
This is \textit{Born's Rule}.

\textit{Remark.} With the frequency $f^{(k)}$ we can associate an operator, $A^{(k)}$, acting on the space of density matrices and defined by 
$$ A^{(k)}(\rho) := \frac{1}{k} \sum_{j=1}^{k} \sum_{\xi_j,...,\xi_{j+l-1}} f(\xi_j,...,\xi_{j+l-1}) \text{  } \pi_{\xi_{j+l-1}}
(t_{j+l-1})\cdot \cdot \cdot \pi_{\xi_j}(t_j) \text{   } \rho \text{   }\pi_{\xi_j}(t_j) \cdot \cdot \cdot \pi_{\xi_{j+l-1}}(t_{j+l-1}),$$
where $\rho$ is an arbitrary density matrix. Viewed as an operator acting on density matrices, an ``observable at infinity'' is obtained by letting $k$ tend to $\infty$ in the above expression. In concrete examples of quantum systems $S$, it is often possible to identify an ``observable at infinity'' with an explicit operator, 
$\mathcal{N}$, (with $|\sigma(\mathcal{N})| = N$), acting on the Hilbert space $\mathcal{H}_{S}$ of the system that has the property that it commutes with time evolution; i.e., $\mathcal{N}(t)\equiv U(t_0,t)\mathcal{N}U(t,t_0)=\mathcal{N}$ \textit{independent} of time $t$. One then speaks of a \textit{non-demolition measurement} of the quantity represented by the operator $\mathcal{N}$.

To conclude this section, we remark that, in practice, the algebra of observables at infinity is most often \textit{trivial} (i.e., it consists of multiples of the identity), because most ``facts'' about a concrete quantum system $S$ tend to actually \textit{vary with time}, and the revelation of such facts, while usually accomplished through long sequences of indirect projective measurements, does \textit{not} correspond to a ``non-demolition measurement'', in the strict sense of the word. In this situation, it remains meaningful to attempt to construct a suitable decomposition of the support of the measure $d\mu_{\omega}$ into subsets corresponding to ``facts'' that are only valid (with \textit{very high} probability) during some long, but finite interval of time. To be concrete, we define functions, $f^{(k,k+r)}$, where $r=0,1,2,...$ is a coarse-grained ``time variable'', as follows:
\begin{equation}\label{2.21}
f^{(k,k+r)}:= \frac{1}{r}\sum_{j=  k +1}^{k+ r}f_{j}(\underline{\xi}),
\end{equation}
see \eqref{2.10}. We introduce sets
\begin{equation} \label{2.22}
\Xi_{\alpha}^{(k,k+r)}( \underline{\varepsilon}):= \lbrace \underline{\xi} \mid  \text{ } \vert f^{(k,k+r)}(\underline{\xi}) - m_{\alpha} \vert < \varepsilon_{r} \rbrace,
\end{equation}
where $m_{\alpha}$ is chosen suitably -- it replaces the quantity 
$f^{(\infty)}(\underline{\xi})|_{\Xi^{\infty}_{\alpha}}$, 
$\alpha=1,...,N$ in Eq. \eqref{2.16} -- and the sequence $(\varepsilon_n)_{n=1}^{\infty}$ tends to $0$, as 
$n\rightarrow \infty$, as above. We assume that 
$$ \text{min}_{\alpha_1\not= \alpha_2}\vert m_{\alpha_1}-m_{\alpha_2} \vert \geq \kappa >0.$$
Clearly, 
$$\Xi_{\alpha_1}^{(k,k+r)} ( \underline{\varepsilon}) \cap \Xi_{\alpha_2}^{(k,k+r)} ( \underline{\varepsilon})  = 
\emptyset$$
if $r$ is so large that $\varepsilon_{r} < \frac{\kappa}{2}$. Next, we introduce the sets
$$\Xi_{\alpha_1,....,\alpha_p}(r;\underline{\varepsilon}):= 
\overset{p}{\underset{j=1}{\bigcap}} \text{ }\Xi_{\alpha_j}^{(j r-r, jr)}( \underline{\varepsilon}).$$
We say that the sequence $(\alpha_1,...,\alpha_p)$ represents a \textit{``history of (plausible) facts''} iff, for a suitable choice of positive integers $r$ and $p\geq1$, the set 
$\cup_{\alpha_1,...,\alpha_{p}}\Xi_{\alpha_1,....,\alpha_p}(r;\underline{\varepsilon})$ has essentially full 
\mbox{$\mu_{\omega}$- measure}, i.e.,
\begin{equation} \label{2.23}
\mu_{\omega}\big(\bigcup_{\alpha_1,...,\alpha_{p}}\Xi_{\alpha_1,....,\alpha_p}(r;\underline{\varepsilon})\big) \geq 1 - \delta_{p,r},
\end{equation}
where $\delta_{p ,r}$ is so small that, for all practical purposes, it is indistinguishable from $0$. If \eqref{2.23} holds then, with probability very close to $1$, \textit{one} of the histories 
$(\alpha_1,...,\alpha_p)$ will take place,
and the probability to observe the history $(\alpha_1,...,\alpha_p)$ is predicted to be given by
\begin{equation} \label{2.24}
\mathcal{P}_{\omega}(\alpha_1,...,\alpha_p) = \mu_{\omega}(\Xi_{\alpha_1,....,\alpha_p}
(r;\underline{\varepsilon})),
\text{   } \text{with } \sum_{\alpha_1,...,\alpha_p}\mathcal{P}_{\omega}(\alpha_1,...,\alpha_p)\approx 1,
\end{equation}
which is a generalization of Born's Rule, see \eqref{2.20}.

In Section \ref{Section4}, we consider ``interesting'' physical quantities of a quantum system $S$ that vary very slowly in time and can be measured indirectly through long, but strictly \textit{finite} sequences of projective measurements of some rather ``uninteresting'' quantities.

To return to the analogy between the quantum systems discussed above and one-dimensional lattice gases (or spin chains) alluded to in the introduction, we point out that ``time'' in the quantum systems becomes ``space'' in the lattice gases , measurement protocols, $\underline{\xi}$, are analogous to configurations of (various species of) particles, the measures $\mu_{\omega}$ are analogous to Gibbs equilibrium measures of lattice gases, the frequencies $f^{(\infty)}$ are analogous to particle densities, with $\alpha$ playing the role of a chemical potential, Theorem \ref{2.2.1} is analogous to results concerning the equivalence of the canonical and the grand-canonical ensemble in the thermodynamic limit, $k \rightarrow \infty$, of lattice gases.
In the theory of lattice gases, it is well known that the Gibbs measures corresponding to different particle densities are mutually singular, (i.e., have disjoint supports). The material discussed between Eqs. \eqref{2.21} and \eqref{2.24} would correspond to a lattice gas with a space-dependent chemical potential
and hence a spatially variable particle density.

 \subsection{Exchangeable measures on the space of measurement protocols} \label{exch}
We conclude Section 2 by considering a special case of the general theory of non-demolition measurements developed in the last two subsections. We assume that the probabilities of results of successive projective measurements of observables in $\mathcal{O}_S$ are independent of the order in which these results are observed. To be more concrete, we imagine that, at each time of measurement, the quantities corresponding to the operators in $\mathcal{O}_S$ are measured with the help of some ``probe'' temporarily interacting with the system, and that the probes used at different measurement times  are different from each other and are prepared and measured independently of each other. A specific model of this type has been described in the Introduction and will be studied in more detail in the next section.
Mathematically, the above assumption is translated into the following property of the measures 
$\mu_{\omega}$:
\begin{equation} \label{2.25}
\mu_{\omega}(\xi_{\pi(1)},...,\xi_{\pi(k)}) \text{   } \text{is } independent \text{    }\text{of the permutation   } \pi,
\end{equation}
for arbitrary permutations, $\pi$, of $\lbrace 1,...,k\rbrace$, for all $k=2,3,...$ .
Measures satisfying \eqref{2.25} are called ``exchangeable''. Such measures have been characterized completely by de Finetti and his followers \cite{Finetti,Hewitt,Aldou}, who have shown that such measures can be written as convex combinations of product measures. Thus, if $\mu_{\omega}$ is exchangeable then there exists a measure space 
$\Xi_{\infty}$, a measure $dP_{\omega}$ on $\Xi_{\infty}$ and probability distributions $p(\cdot|\nu)$ on
$\sigma_S$,
$\nu \in \Xi_{\infty}$, such that
\begin{equation} \label{2.26}
\mu_{\omega}(\xi_1,...,\xi_k) = \int_{\Xi_{\infty}} dP_{\omega}(\nu)\text{    } \prod_{j=1}^{k} p(\xi_j|\nu).
\end{equation}
This equation is a concrete instance of the general representation \eqref{2.9}. It follows from Kolmogorov's zero-one law that $\Xi_\infty$ appearing in the decomposition is indeed the spectrum of the tail algebra. For exchangeable measures, we introduce the frequencies
$$f_{\xi}^{(k)}(\underline{\xi}):= \frac{1}{k} \sum_{j=1}^{k} \delta_{\xi_j,\xi}, \qquad \text{and  } \text{   }\phi_{\nu, \xi}^{(k)}(\underline{\xi}) :=\sqrt {k} \big( f_{\xi}^{(k)}(\underline{\xi}) - p(\xi \vert \nu) \big). $$
Using the Law of Large Numbers and the Central Limit Theorem in the form due to Berry and Ess\'{e}en, 
see, e.g., \cite{Feller}, it is straightforward to prove an explicit (quantitative) version of Theorem 2.2.1.

In subsequent sections, we study families of models of systems $S$ giving rise to measures $\mu_{\omega}$ that are exchangeable for \textit{arbitrary} states $\omega$ of $S$. Such models have the advantage of being quite easy to analyze mathematically.

\section{A family of simple models of quantum systems} \label{Section3}

\subsection{Description of non-demolition measurements}

\label{des_non_dem}

A general mathematical formalism to describe the statistics of measurement outcomes and the reduced evolution in an indirect measurement was developed by Kraus \cite{Kraus}, (see also \cite{Holevo}). Here we briefly recall some elements of this formalism. We consider a system $S$ composed of a subsystem $\overline{P}$ of interest to an experimentalist and a subsystem $E$ consisting of a measurement device. We suppose that $\mathcal{H}_S= \mathcal{H}_{\overline{P}} \otimes  \mathcal{H}_{E}$. The reduced time evolution of the subsystem $\bar{P}$, conditioned upon a measurement event $\wp \subset \sigma_S$, is encoded in a completely positive map $\Phi_{*}(\wp)$ acting on $\mathcal{B}(\mathcal{H}_{\overline{P}})$. To give rise to a probability measure on a space of measurement events this family of maps must be countably additive, (in particular, for disjoint sets 
$\wp_1$ and $\wp_2$, $\Phi_*(\wp_1 \cup \wp_2) =\Phi_*(\wp_1) + \Phi_*(\wp_2)$), and  
$\Phi_*(\sigma_S)$ must be the identity map on $\mathcal{B}(\mathcal{H}_{\overline{P}})$, (i.e.,
$\Phi_*(\sigma_S)[\mathds{1}] = \mathds{1}$). The maps $\Phi_*(\cdot)$ describe the evolution of ``observables'' of the subsystem $\overline{P}$ in the presence of measurement events (i.e., events happening in the device $E$); the evolution of states of $\overline{P}$ is described by
 dual maps, $\Phi(\wp) \equiv (\Phi_*(\wp))^*$.

For simplicity we assume that $\sigma_S$ is finite and that $\dim(\mathcal{H}_{\overline{P}})< \infty$. Hence $\Phi_*(\cdot)$ is completely determined by the maps $\Phi_*(\{\xi\}) =: \Phi_{* \xi}$, with $\xi \in \sigma_S$; in particular, $\Phi_*(\sigma_S) = \sum_{\xi \in \sigma_S} \Phi_{*\xi}$. For a given state, $\omega(\cdot) = \tr(\rho^{(0)}\, (\cdot))$, of $\overline{P}$, the probability to observe a measurement result $\xi$ is given by $\omega(\Phi_{*\xi}[\mathds{1}])$, and the reduced non-normalized posterior state of $\overline{P}$, given that $\xi$ has been observed, is 
$\Phi_\xi [\rho^{(0)}]$. For a sequence of independent  identical probes, as discussed in Section \ref{exch},  the probability to observe the   measurement protocol $\underline{\xi}^{(k)} = (\xi_1, \, \dots,\,\xi_k)$ is given by 
\begin{equation}
\label{3.1.1}
\mu_\omega(\underline{\xi}^{(k)}) = \omega(\Phi_{* \xi_1} \circ \dots \circ \Phi_{* \xi_k}[ \mathds{1}]),
\end{equation}
see  \eqref{LSW}.
The corresponding reduced state, $\rho^{(k)}(\underline{\xi}^{(k)})$, of $\overline{P}$ is then given by
\begin{equation}
\label{3.1.2}
\rho^{(k)}(\underline{\xi}^{(k)}) = \frac{\Phi_{\xi_k} \circ \dots \circ \Phi_{\xi_1} [\rho^{(0)}]}{\tr(\Phi_{\xi_k} \circ \dots \circ \Phi_{\xi_1} [\rho^{(0)}])}. 
\end{equation}
Our goal is to find the tail algebras for all such i.i.d measurement processes, i.e., to determine the ``facts'' emerging, in a given experiment, from projective observations of very (infinitely) many identical, independent probes that have interacted with the subsystem $\overline{P}$. In this paper we wish to determine these tail algebras in the setting described in Sections 1 and 2.

In a non-demolition measurement, the mere presence of a measurement device does not affect the probabilities of sequences of measurement results. Mathematically, this translates into the condition
\begin{equation}
\label{3.1.3}
\sum_{\xi_{j} \in \sigma_S} \mu_\omega(\xi_1, \dots,\xi_{j-1},\xi_j,\xi_{j+1}, \dots , \xi_k) = \mu_\omega(\xi_1, \dots,\xi_{j-1}, \xi_{j+1}, \dots , \xi_k).
\end{equation}
 Upon inspection of Eq.~(\ref{3.1.1}), we see that a suitable assumption that guarantees that this condition holds is
\begin{equation}
\label{3.1.4}
\Phi_{*\xi} \circ \Phi_{*\xi'}  = \Phi_{*\xi'} \circ  \Phi_{*\xi}, \quad \mbox{for all} \quad \xi,\,\xi' \in \sigma_S.
\end{equation}
(A more detailed discussion of such hypotheses will be presented elsewhere). 
Condition \eqref{3.1.4} immediately implies that the measures $\mu_{\omega}$ are exchangeable, see  \eqref{2.25}, and hence $\mu_{\omega}$ has a de Finetti decomposition; see \eqref{2.26}. Assuming that condition \eqref{3.1.4} is satisfied,  we  can fully describe the ``algebra of pointer observables''  under the  technical assumption that  $\Phi(\sigma_S)$ has a faithful stationary state. In this particular case, we show in Lemma~\ref{lem:ins2} that 
\begin{equation}
\label{3.1.5}
[\Phi_{*\xi}[\mathds{1}], \Phi_{*\xi'}[\mathds{1}]] = 0,  \quad \mbox{for all} \quad \xi,\,\xi' \in \sigma_S.
\end{equation}
We denote by $\{\Pi_{\nu} \vert \Pi_{\nu} \in B(\mathcal{H}_{\overline{P}}), \nu \in \Xi_{\infty} \}$, the joint spectral projections of the commuting family 
$\{\Phi_{*\xi}[\mathds{1}]\}_{\xi \in \sigma_S}$. The conditional probabilities $p(\xi|\nu)$ appearing in the decomposition (\ref{2.26}) are the eigenvalues of $\Phi_{* \xi}$ corresponding to the eigenprojections 
$\Pi_\nu$, i.e., 
\begin{equation}
\label{3.1.8}
\Phi_{* \xi} [\mathds{1}] = \sum_{\nu \in \Xi_\infty} p(\xi|\nu) \Pi_\nu,
\end{equation}
and the probabilities $p(\cdot \vert \nu)$ on $\sigma_S$ are mutually distinct. Eq. (\ref{2.26}) then takes the form
\begin{equation}
\label{3.2}
\mu_{\omega}(\underline{\xi}) = \sum_{\nu \in \Xi_\infty} \omega(\Pi_\nu)\, \mu_{\omega}(\underline{\xi} \vert \nu), \quad \text{with } \quad \mu_{\omega} (\xi_1,...,\xi_k \vert \nu)= \prod_{j=1}^{k} p(\xi_j \vert \nu).
\end{equation}
 The tail algebra is  isomorphic to the commutative algebra generated by the orthogonal projections 
 $\lbrace \Pi_{\nu} \rbrace_{\nu \in \Xi_{\infty}}$, i.e., by all operators of the form  $\sum_\nu f(\nu) \Pi_\nu$,  where $f : \Xi_{\infty} \rightarrow \mathbb{R}$ is an arbitrary bounded measurable function on $\Xi_{\infty}$. 
In \cite{Mass}, Maassen and K{\"u}mmerer have considered the special case where
\begin{equation}
\label{3.1.9}
\Phi_\xi[ \rho] = C_\xi \rho C_\xi^*,
\end{equation}
$C_\xi$  being the operators on $\mathcal{H}_{\overline{P}}$ that appear in a Kraus decomposition,
\begin{equation}
\label{3.1.7}
\Phi(\sigma_s) [\rho] = \sum_{\xi \in \sigma_S} C_\xi \rho C^*_\xi, \quad \sum_{\xi \in \sigma_S} C^*_\xi C_\xi = \mathds{1}.
\end{equation}
A  special feature of this choice is that each map $\Phi_\xi$ maps pure states to pure states. The family 
$\{C_{\xi} \}_{\xi \in \sigma_S}$ of operators satisfies \eqref{3.1.4} and \eqref{3.1.8},  provided  that $C_\xi = \sum_{\nu \in \Xi_{\infty}} c_\xi(\nu) \Pi_\nu$, with $|c_{\cdot}(\nu)| \neq |c_{\cdot}(\nu')|$ if $\nu \neq \nu'$. A straightforward calculation then shows that
$$
\Phi_{* \xi}[\mathds{1}] = \sum_{\nu \in \Xi_{\infty}} p(\xi|\nu) \Pi_\nu, \quad \text{where    } p(\xi | \nu) = |c_\xi(\nu)|^2.
$$
The normalization condition in  (\ref{3.1.7}) ensures that $p(\cdot|\nu)$ is a probability distribution.
In a non-demolition experiment of the kind described in the Introduction, the coefficients $c_\xi(\nu)$ can be interpreted as transition amplitudes for a probe (always prepared in the same initial state, $\psi_0$) on which a projective measurement is performed after it has interacted with $\overline{P}$, with the measurement result corresponding to a final state $\psi_{\xi}$. The Hamiltonian describing the time evolution of the probe interacting with the system $\overline{P}$ depends on the point $\nu \in \Xi_{\infty}$; i.e., the joint Hamiltonian of the system $\overline{P}$ and a single probe has the form
$
H = \sum_{\nu \in \Xi_{\infty}} \Pi_\nu \otimes H_\nu,
$
which leads to a joint time evolution (over one measurement cycle) of the form
\begin{equation} \label{Uhu}
U = \sum_{\nu \in \Xi_{\infty}} \Pi_\nu \otimes U_\nu.
\end{equation}
The coefficients $c_{\xi}(\nu)$ are then given by $c_{\xi}(\nu) = \langle \psi_{\xi}, U_{\nu} \psi_{0} \rangle$, where $\psi_{0}$ is the initial state of the probe and $\psi_{\xi}$ is its state after a projective measurement on the probe sub-system has yielded the result $\xi$.

\subsection{Emergence of facts in non-demolition measurements} \label{nondem}
In this section we proceed with the analysis of the measures $\mu_{\omega}$ given in Eqs.~(\ref{2.26}) and (\ref{3.2}), under the additional assumption that the spectrum $\Xi_\infty$ consists of finitely many points and, hence, can be identified with a subset of $\mathbb{R}$. We assume that   $\sigma_S$  is  countable and allow $\mathcal{H}_{\overline{P}}$ to be infinite dimensional.  It is well known that the product measures appearing in Eqs. ~(\ref{2.26}) and (\ref{3.2}) satisfy large deviation estimates. For every fixed $\xi \in \sigma_S$, we introduce the random variable 
\begin{equation}\label{alf3rep} 
f_{\xi}^{(k)}( \underline{\xi}) \equiv f_{\xi}^{(0,k)}(\underline{\xi}) : = \frac{1}{k} \# \Big \{ j \in \{1, \cdots, k \} \: \Big | \: \xi_j = \xi  \Big \}  
\end{equation}
on $(\Xi,\Sigma)$. Given $\nu \in \Xi_{\infty}$, the empirical measures $\sum_{\xi \in \sigma_S} f_{\xi}^{(k)} \delta_{\xi}$  satisfy large deviation estimates and converge to $\sum_{\xi \in \sigma_S} p(\xi \vert \nu) \delta_{\xi}$ exponentially fast; see Theorem  \ref{bs}, below. Before stating this result more precisely, we introduce some notation and recall some preliminary facts.

We denote by $\mathcal{M}(\sigma_S)$ the polish space of probability measures on 
$\sigma_S$; (we note, in passing, that $\sigma_S$ is a polish space, because every countable compact Hausdorff space is polish). Convergence in  $\mathcal{M}(\sigma_S)$  is equivalent to weak convergence, i.e., to convergence against bounded continuous functions on $\sigma_S$.  Let $p= \sum_{\xi \in \sigma_S} p(\xi) \delta_{\xi} \in \mathcal{M}(\sigma_S)$  be an arbitrary probability measure on $\sigma_S$. We denote by 
  \begin{equation} \label{pnu}
  p_{\nu}= \sum_{\xi \in \sigma_S} p(\xi \vert \nu) \delta_{\xi}
  \end{equation}
   the probability measure on $\sigma_S$ with density $p(\xi \vert \nu)$. The \textit{relative entropy} of the probability measure $p$ with respect to the measure $p_{\nu}$ is defined as 
\begin{align} \label{ent}
I_{p_{\nu}}(p) : = \begin{cases} \sum_{\xi \in \sigma_S} p(\xi) \ln \frac{p(\xi)}{p(\xi\vert \nu)} , & \text{if} \:  p \ll p_{\nu}, \: 
\\ \infty, & \text{otherwise}.  
\end{cases}
\end{align}
With an arbitrary subset $K \subset \mathcal{M}(\sigma_S) $ of probability measures on $\sigma_S$ we associate the quantity
\begin{align}\label{entGamma}
I_{p_{\nu}}(K) : =  \inf_{p \in K} I_{p_{\nu}}(p).
\end{align}
  It is easy to see that $I_{p_{\nu}}  $ is a convex function of $p$ and that $I_{p_{\nu}}(p) \geq 0  $, for every $p \in \mathcal{M}(\sigma_S)$, with equality iff  $ p=p_{\nu} $. The following theorem is  well known; (see, e.g., \cite{dembo,Ellis}).
  
\begin{theorem}[Boltzmann, Sanov] \label{bs} 
For any closed subset $K$ of $\mathcal{M}(\sigma_S)$ not containing $p_{\nu}$, one has that  $ I_{p_{\nu}}(K)>0$, and there is a constant $C_{K}> 0$ such that
\begin{equation} \label{largee}
\mu_{\omega}\Big(\big \{ \underline \xi \: \big | \:  \sum_{\xi' \in \sigma_S} f_{\xi'}^{(k)}(\underline{\xi}) \delta_{\xi'}  \in K \big \}  \Big \vert \nu \Big) \leq C_{K} \exp(-k  I_{p_{\nu}}(K)/2),
\end{equation}
for all $k \in \mathbb{N}$.
\end{theorem}

The construction leading to Eq.~(\ref{3.2}) ensures that the probability distributions $p_{\nu}$ and $p_{\nu'}$ are distinct, and  hence the measures  $\mu_{\omega}(\cdot \vert \nu)$ and $\mu_{\omega}(\cdot \vert \nu')$ are mutually singular.  We can then decompose $\Xi$ into a collection $(\Xi_{\nu})_{\nu \in \Xi_{\infty}}$ of measurable subsets of $\Xi$ such that $\mu_{\omega}(\Xi_{\nu'} \vert \nu)=\delta_{\nu,\nu'}$. The proof of  the mutual singularity of the product measures $\mu_{\omega}(\cdot \vert \nu)$ and $\mu_{\omega}(\cdot \vert \nu')$  is straightforward (involving a \textit{lim inf} of sequences of sets). In order to keep this paper as self-contained as possible, we sketch it in an appendix, Sect. 5. As already discussed in the Introduction and in Section 2, Theorem \ref{bs} enables an experimentalist in the lab to measure properties of the system $\overline{P}$ corresponding to functions on $\Xi_{\infty}$ indirectly, and hence to determine 
$\nu \in \Xi_{\infty}$, by recording empirical frequencies of long sequences of outcomes of projective measurements of (possibly very uninteresting) quantities in $\mathcal{O}_S$.

\subsection{Decoherence and ``purification''} \label{sec311}

In the previous subsection, we have shown that, in non-demolition experiments, facts $\nu \in \Xi_\infty$ (identified with a subset of $\mathbb{R}$) emerge exponentially fast. In this subsection, we are interested in studying the time evolution of the reduced density matrix describing the state of $\overline{P}$, as projective measurements of quantities in $\mathcal{O}_S$ are carried out, with results $\xi_1,..., \xi_k$ belonging to 
$\sigma_S$, with $k=1,2,3,...$. In particular, we will prove that, under suitable hypotheses, the reduced density matrix of $\overline{P}$ has a limit, as 
$k \rightarrow + \infty$.

At the level of generality chosen at the beginning of Section~\ref{des_non_dem}, the limit of these reduced density matrices may not exist, because \eqref{3.1.8} does not sufficiently constrain the time evolution restricted to the range of $\Pi_\nu$.
We avoid this difficulty by assuming that the time evolution of the system consisting of $\overline{P}$ and one probe has the form given in Eq. \eqref{Uhu}. Thus, in this section we consider maps $\Phi_\xi$ given by Eq.~(\ref{3.1.9}), with (cf. text below Eq.~(\ref{3.1.7}))
$$
C_\xi = \sum_{\nu \in \Xi_\infty} c_\xi(\nu) \Pi_\nu.
$$
In summing over  $\Xi_\infty$  we tacitly assume that the measures $p_\nu$ (see  \eqref{pnu}) are mutually distinct. The dynamics, constrained in this way, has the property that the maps $\Phi_{*\xi}$ act on operators of the form $\Pi_{\nu}(\cdot) \Pi_{\nu'}$ as multiplication operators, and we have that
\begin{equation}
\label{3.2.1}
\Pi_\nu (\Phi_\xi [\rho]) \Pi_{\nu'} = c_{\xi}(\nu)  \Pi_\nu  \rho \Pi_{\nu'} \overline{c_{\xi}(\nu')},
\end{equation}
for an arbitrary density matrix $\rho$ on $\mathcal{H}_{\overline{P}}$.
Next, we state and prove a result of \cite{Mass} rediscovered in \cite{BaBe}. Our proof is significantly different from those in \cite{Mass, BaBe}; (it shows that the result is a nearly immediate corollary of Theorem \ref{bs}).  In the following, $\| \cdot\|$ denotes either the  operator norm or the trace-norm, the statements being correct for both choices.

\begin{theorem} \label{312}
For $\nu \not= \nu'$,
\begin{equation} \label{dec}
\| \Pi_{\nu} \rho^{(k)}(\underline{\xi}^{(k)}) \Pi_{\nu'} \| \rightarrow 0,  \qquad \mu_{\omega} \text{-}\text{   } a.s.,
\end{equation}
as $k\rightarrow \infty$, where $\rho^{(k)}(\underline{\xi}^{(k)})$ has been defined in \eqref{3.1.2}.
Furthermore, there exists a random variable $\Theta: \Xi \rightarrow \Xi_{\infty}$ such that 
\begin{equation} \label{purififi}
 \Big \|  \rho^{(k)}(\underline{\xi}^{(k)})  -  \frac{ \Pi_{\Theta} \rho^{(0)}   \Pi_{\Theta} }{ \mathrm{Tr}(\Pi_{\Theta}   \rho^{(0)} \Pi_{\Theta}))} \Big \|  \rightarrow 0, \qquad \mu_{\omega} \text{-} \text{ }  a.s.,
\end{equation}
as $k\rightarrow \infty$.
The probability that $\Theta=\nu$ is   equal to $P_{\omega}(\nu) = \mathrm{Tr}( \rho^{(0)} \Pi_{\nu})$.
\end{theorem}

\begin{proof}

Eq.~(\ref{3.2.1}) yields a recurrence relation for the density matrix $\rho^{(k)}$,

\begin{equation}\label{recu}
 \Pi_{\nu} \rho^{(k)} (\underline{\xi}^{(k)}) \Pi_{\nu'} = \frac{ c_{\xi_k}(\nu) \Pi_{\nu} \rho^{(k-1)} (\underline{\xi}^{(k-1)}) \Pi_{\nu'}  \overline{c_{\xi_k}(\nu')}}{ \underset{\nu \in \Xi_{\infty}}{ \sum}  \vert  c_{\xi_k}(\nu)  \vert^2   \mathrm{Tr}(\Pi_{\nu} \rho^{(k-1)} (\underline{\xi}^{(k-1)})) },
\end{equation}
which, by iteration, implies that
\begin{equation} \label{decsn}
\Pi_{\nu} \rho^{(k)} (\underline{\xi}^{(k)}) \Pi_{\nu'}= \frac{1}{ 
\mu_{\omega}(\xi_1,...,\xi_k) } \text{ } \Pi_{\nu} \rho^{(0)}  \Pi_{\nu'} \text{ }  \prod_{i=1}^{k}  c_{\xi_i}(\nu) \overline{ c_{\xi_i}(\nu')}.
\end{equation}
Multiplying both sides by by $\mu_{\omega}(\xi_1,...,\xi_k)$, taking norms, and summing over all possible 
$\xi_1,...,\xi_k$, we find that
\begin{equation} \label{dec}
\mathbb{E} \|  \Pi_{\nu} \rho^{(k)} (\underline{\xi}^{(k)}) \Pi_{\nu'} \| \leq \| \Pi_{\nu} \rho^{(0)}  \Pi_{\nu'}  \|  (\sum_{\xi \in \sigma_S} \vert  c_{\xi}(\nu) \vert   \text{ }\vert  c_{\xi}(\nu') \vert )^k \leq \| \Pi_{\nu} \rho^{(0)}  \Pi_{\nu'}  \|  (\delta_{\nu \nu'})^k,
\end{equation}
for some $0<\delta_{\nu \nu'}<1$. The factor $(\delta_{\nu \nu'})^k$ on the right side is obtained from the Cauchy-Schwarz inequality  and the non-degeneracy assumption,  $p_{\nu} \neq p_{\nu'}$ for $\nu \not= \nu'$. The exponential decay of the expected value in Eq. \eqref{dec} implies almost sure convergence -- a consequence of the Borel-Cantelli lemma. 

To prove \eqref{purififi}, we observe that from the fact that the measures $\mu_{\omega}(\cdot \vert \nu)$ are mutually singular, as $\nu$ ranges over $\Xi_{\infty}$, one may deduce that  the sets $\Xi_{\nu}$ (introduced at the end of Section \ref{nondem}) have the property that, for $\nu \neq \nu'$,
\begin{equation} \label{3.15}
\|   \Pi_{\nu} \rho^{(k)} (\underline{\xi}^{(k)}) \Pi_{\nu}  \|\chi_{\Xi_{\nu'}}  \rightarrow 0, \qquad \text{as  } \text{   } k\rightarrow \infty,
\end{equation}
$\mu_{\omega}$- a.e.. This property is proven in the appendix (see subsection 5.1.2), where we will also show that 
\begin{equation}  \label{3.16}
   \|  \mathrm{Tr}(  \Pi_{\nu} \rho^{(0)} )  \Pi_{\nu} \rho^{(k)} (\underline{\xi}^{(k)}) \Pi_{\nu}  -\Pi_{\nu} \rho^{(0)} \Pi_{\nu}  \| \chi_{\Xi_{\nu}}(\underline{\xi}) \rightarrow 0, \qquad  \mu_{\omega}\text{-} \text{ a.e.},
\end{equation}
as $k \rightarrow \infty$.
Let $\Theta$ be the random variable defined by $$\Theta(\underline{\xi}) =\sum_{\nu \in \Xi_{\infty}} \chi_{\Xi_{\nu}}(\underline{\xi})  \text{ }\nu.$$ Using \eqref{dec}, \eqref{3.15}, \eqref{3.16}, the triangle inequality and the fact that $\mu_{\omega}(\cup_{\nu} \Xi_{\nu})=1$, one shows that 
\begin{align*}
\Big \|  \rho^{(k)}(\underline{\xi}^{(k)})   -  \frac{ \Pi_{\Theta(\underline{\xi})} \rho^{(0)}   \Pi_{\Theta(\underline{\xi})} }{ \mathrm{Tr}(\Pi_{\Theta(\underline{\xi})}   \rho^{(0)} \Pi_{\Theta(\underline{\xi})}))} \Big \| 
\end{align*} 
converges to zero, $\mu_{\omega}$- a.e..
\end{proof}

\section{Indirect measurements of physical quantities varying slowly in time} \label{Section4}
\subsection{A perturbative approach to indirect measurements}
At the end of Section \ref{hypo}, we outlined some general ideas concerning indirect measurements or observations of physical quantities that evolve slowly in time; see Eqs. \eqref{2.21} through \eqref{2.24}. In this section, we intend to add some mathematical precision to those ideas by analyzing indirect measurements of a physical quantity evolving slowly in time, for a class of simple models. For this purpose, we develop a \textit{perturbative approach} to the theory of indirect measurements. We consider models whose time-evolutions are perturbatively close to the ones considered in our analysis of non-demolition measurements presented in Section \ref{des_non_dem}. As in Section \ref{des_non_dem}, we  consider models of a quantum system $S$ that is the composition of a subsystem $\overline{P}$ of primary interest with a subsystem $E$ consisting of equipment used to observe $\overline{P}$. 
We  are interested in describing observations of physical properties of the subsystem 
$\overline{P}$, also called ``facts'', using the experimental equipment described by $E$. For this purpose, successive projective measurements of some quantities referring to $E$ are carried out. They yield a sequence of reduced states, $\rho^{(k)}(\underline{\xi}^{(k)})$, on the algebra, $B(\mathcal{H}_{\overline{P}})$, of bounded operators on $\mathcal{H}_{\overline{P}}$ depending on  measurement protocols 
$\underline{\xi}^{(k)}=(\xi_1,...,\xi_k) \in (\sigma_{S})^{\times k}$ of arbitrary length $k< \infty$, (with $\xi_j$ the values of a quantity referring to $E$ measured in the $j^{th}$ projective measurement).

The reduced state of $\overline{P}$, after $k$ projective measurements carried out on $E$, is assumed to be given, recursively, by the density matrix
\begin{equation}\label{Rec}
\rho^{(k)}(\underline{\xi}^{(k)}) = \frac{\Phi_{\underline \xi}^{(k)}[\rho^{(k-1)}(\underline{\xi}^{(k-1)})]}{\mathrm{Tr}(\Phi_{\underline \xi}^{(k)}[\rho^{(k-1)}(\underline{\xi}^{(k-1)})])},
\end{equation}
where, for all  $\underline{\xi}$, the evolution map $\Phi_{\underline \xi}^{(k)}$ only depends on the first $k$ measurement results $\underline{\xi}^{(k)} = (\xi_1,..., \xi_k)$, for all $k \in \mathbb{N}$.  As before,  we assume that  each map $\Phi_{\underline \xi}^{(k)}$ is the dual of a completely positive map and that 
\begin{equation}
\sum_{\xi_k \in \sigma_S} \tr(\Phi_{\underline \xi}^{(k)}[\rho])=\tr(\rho)
\end{equation}
for any trace-class operator $\rho \in \mathcal{B}(\mathcal{H}_{\overline{P}})$. The probability of the measurement protocol $\underline{\xi}^{(k)}$, assuming that the system has been prepared in the state 
$\omega$ corresponding to a density matrix $\rho^{(0)}$, is given by 
\begin{equation}
\mu_{\omega}(\xi_1,...,\xi_k)=\mathrm{Tr}(\Phi_{\underline \xi}^{(k)} \circ \cdots  \circ \Phi_{\underline \xi}^{(1)}[\rho^{(0)}]).
\end{equation}

In contrast to the assumptions required in Section~\ref{des_non_dem}, we allow the commutator appearing in Eq.~(\ref{3.1.5}) to be non-zero, but small, (and the system is usually not autonomous).We compare the evolution described in \eqref{Rec} to one considered in our analysis of non-demolition measurements which it is assumed to be close to. We add a tilde to all quantities referring to non-demolition measurements, as treated in Section \ref{sec311}; namely, we will write
\begin{align}
\tilde{\Phi}_{\xi}[\rho]&:=\sum_{\nu, \nu' \in \Xi_{\infty}} c_{\xi}(\nu) \Pi_{\nu} \rho  \Pi_{\nu'}\overline{c_{\xi}(\nu')} ,\\
\tilde{\mu}_{\omega}(\xi_1,...,\xi_n)&:=\mathrm{Tr}( \tilde \Phi_{  \xi_n} \circ  \cdots \circ  \tilde \Phi_{  \xi_1}[\rho^{(0)}]).
\end{align}
To avoid confusions, the trace norm of a trace-class operator $\rho$ is denoted by $\| \rho \|$, and  the operator norm of an operator $A$ in $\mathcal{B}(\mathcal{H}_{\overline{P}})$ is denoted by $\|A \|_{\mathrm{op}}$. The operator norm of a map $\Psi$ (= $\tilde{\Phi}_{\xi}$ or $\Phi_{\underline \xi}^{(k)}$), viewed as an operator on the Banach space of trace-class operators, is denoted by $\| \Psi \|$. In this section, we always assume that $\sigma_S$ and $\Xi_{\infty}$ are finite point sets. We propose to study the dynamics of the reduced system $\overline{P}$ under the following assumptions:

\begin{assumption}\label{Ass:m1}
We assume that there exist constants $d_1 \in [0, 1)$ and  $d_2  \in (d_1,1] $ such that, for all $n \in \mathbb{N}$, for every $\xi \in \sigma_S$ and  every $\underline {\xi}$ with $  \xi_k = \xi $, 
\begin{enumerate}
\item[(i)] $\|  \Phi_{\underline \xi}^{(k)} - \tilde \Phi_{ \xi} \|
 \leq d_1 \|  \tilde \Phi_{ \xi} \|$,
\item[(ii)] $
\mathrm{Tr}( \tilde \Phi_{\xi} \rho) \geq d_2   \|   \tilde  \Phi_{\xi} \| $, for all density matrices $\rho$ on $\mathcal{H}_{\overline{P}}$.
\end{enumerate}
\end{assumption}

\noindent The first assumption concerns the smallness of the difference between the actual evolution and one corresponding to a non-demolition measurement, while the second assumption implies  that the a-priori probability of  any outcome of a projective measurement of a quantity referring to $E$ (used in a 
non-demolition measurement of a quantity of $\overline{P}$) never vanishes.  
Items (i)  and (ii) imply that
\begin{equation}
\label{eq:m1}
 \mathrm{Tr}(  \Phi_{\underline \xi}^{(k)} [\rho]) \geq  \underset{=:d}{\underbrace{(d_2 - d_1)}}\|  \tilde \Phi_{\xi_k} \|,
\end{equation}
 for an arbitrary density matrix $\rho$ and any $k$. Without loss of generality, we also assume that $d=d_2 - d_1<1$. 
 Before stating the main results of this section, we give some examples of maps satisfying Assumption \ref{Ass:m1}.
 \begin{example}[illustrating (i)]
Let $(H^{(k)})_{k=1}^{\infty}$ be a sequence of Hamiltonians with $ \| H^{(k)}\|_{\mathrm{op}} \leq d_1/2$, for all $k$. Then (i) holds for a perturbation dynamics 
$ \Phi^{(k)}_{\underline \xi} = \tilde{\Phi}_{\xi_k} \exp(-i Ad_{H^{(k)}})$.
\end{example}

\begin{example}[illustratig (ii)]
Let $ \tilde \Phi_\xi = \upsilon_\xi \text{  } {id} + \Upsilon_\xi$, where $\upsilon_\xi$ is a probability distribution on $\sigma_S$ and $\Upsilon_\xi$ is a family of maps satisfying $\sum_\xi \Upsilon_\xi = 0$. Then $ \tilde  \Phi_\xi$ satisfies (ii) with the constant 
$$d_2=  \sup_{\xi \in \sigma_S} \left(1 - 2 \frac{\|\Upsilon_\xi\| }{\upsilon_\xi + \| \Upsilon_\xi \| }\right),$$
provided $\upsilon_\xi >  \| \Upsilon_\xi \| $, uniformly in $\xi$.  If the ratio of $\| \Upsilon_\xi \|$ and $\upsilon_\xi$ is small, for all $\xi$, then $d_2$ is close to $1$.
\end{example}

\noindent In studying the first example we use the  inequality $|1-e^{ix}| \leq  \vert x \vert$  and conclude that 
\begin{align*}
\| \Phi^{(k)}_{\underline \xi} - \tilde \Phi_{\xi_k} \|  &\leq \|  \tilde \Phi_{\xi_k} \| \, \| 1 - \exp(-i Ad_{H^{(k)}})\|  \\
						&\leq 2 \| \tilde \Phi_{\xi_k} \|   \| H^{(k)}\|_{\mathrm{op}} ,
\end{align*}
where we have used that $\|Ad_{H} \| \leq 2\| H \|_{\mathrm{op}} $. To treat the second example, we remark that (when viewed as a  map acting on the space of trace-class operators) $\tilde \Phi_\xi $ satisfies the inequality
$$
(\upsilon_\xi - \| \Upsilon_\xi \| )\cdot id \leq  \tilde \Phi_\xi \leq (\upsilon_\xi + 
\|\Upsilon_\xi\| )\cdot id.
$$
Hence we have that
\begin{align*}
\tr( \tilde \Phi_\xi \rho) &\geq \upsilon_\xi - \| \Upsilon_\xi\|  =  \upsilon_\xi + 
\| \Upsilon_\xi \|  - 2 \| \Upsilon_\xi \|  \\
				& \geq \| \tilde \Phi_\xi \|  (1- 2 \frac{\| \Upsilon_\xi \|}{\upsilon_\xi + \| \Upsilon_\xi \| }),
\end{align*}
for an arbitrary density matrix $\rho$.

\subsection{Trajectories of quantum jumps on $\Xi_{\infty}$}
We denote by 
\begin{equation}
\mathcal{N}=\sum_{\nu \in \Xi_{\infty}} \nu \Pi_{\nu} \in \mathcal{B}(\mathcal{H}_{\overline{P}})
\end{equation}
the operator representing the physical quantity of $\overline{P}$ that we wish to measure indirectly. (It is assumed here that $\Xi_{\infty}$ is a subset of $\mathbb{R}$ and that $\nu$ is a function on $\Xi_{\infty}$ that separates points of $\Xi_{\infty}$. In the example discussed in the Introduction, the operator 
$\mathcal{N}$ is the number operator counting the number of electrons in the component $P$ of the quantum dot $\overline{P}$ close to the conducting channel.) The time evolution of  $\mathcal{N}$ is assumed to be non-trivial under the dynamics corresponding to the maps $\Phi_{\underline{\xi}}^{(k)}$. Nevertheless, one may hope that the measurement protocols $\underline{\xi}^{(k)}$ can be used to track the  values of $\mathcal{N}$ during fairly long, but finite intervals of time (on which the value of $\mathcal{N}$ is constant with very high probability), because  the dynamics determined by the maps $\Phi_{\underline{\xi}}^{(k)}$ is assumed to be  close to the ``non-demolition dynamics'' studied in Section \ref{sec311}; see  Assumption \ref{Ass:m1}. For each $k \in \mathbb{N}$, we  introduce an ``estimator''  $\hat{\mathcal{N}}^{(k,k+r)}(\underline\xi)$ whose value correctly predicts the outcome of a direct measurement of 
$\mathcal{N}$ at a time $\approx   k+r $ with high probability, provided the outcomes, 
$(\xi_{k+1},...,  \xi_{k+r})$, of $r$ projective probe measurements are known, for some ``time constant'' 
$r \in \mathbb{N}$: For every fixed $\xi \in \sigma_S$, we set
\begin{equation}\label{alf3rep} 
f_{\xi}^{(k,k+r)}( \underline{\xi}) : = \frac{1}{r} \# \Big \{ j \in \{k+1, \cdots, k + r \} \: \Big | \: \xi_j = \xi  \Big \},  
\end{equation}
and we define the estimator $\hat{\mathcal{N}}^{(k,k+r)}(\underline\xi)$  by  
\begin{equation}\label{esti}
\hat{\mathcal{N}}^{(k,k+r)}(\underline\xi) := \underset{\nu \in \Xi_{\infty} }{\mathrm{argmin} } \text{ }  I_{p_{\nu}}\big ( \sum_{\xi \in \sigma_S} f_{\xi}^{(k,k+r)}( \underline{\xi}) \text{ } \delta_{\xi}  \big),
\end{equation}
for any fixed choice of argmin in case there does not exist a unique minimizer.  We remind the reader that the  measure   $p_{\nu}$ has been defined in \eqref{pnu}, and that the coefficients $p(\xi \vert \nu)$ are those that have been introduced in the ``non-demolition model'' discussed in Section \ref{sec311}, i.e., $p(\xi \vert \nu)=\vert c_{\xi}(\nu) \vert^2.$ The level sets of the estimator $\hat{\mathcal{N}}^{(k,k+r)}(\underline\xi)$, 
\begin{equation}
\Xi_{\hat{\mathcal{N}}^{(k,k+r)} =\nu}:  =\Big \{  \underline \xi  \in \Xi \mid  
\hat{\mathcal{N}}^{(k,k+r)}(\underline\xi) = \nu \Big \},
\end{equation}
yield a partition of the space  $\Xi$ of infinitely long measurement protocols.  
The probability of an error in the prediction by the estimator of the outcome of a direct measurement of 
the value of $\mathcal{N}$ at time $ k+r$ is then given by
\begin{align} \label{plk}
\epsilon^{(k,k+r)}(\nu) :=  \sum_{\underline{\xi} \in \Xi_{\hat{\mathcal{N}}^{(k,k+r)} =\nu}} \mathrm{Tr} \big ((1-\Pi_{\nu}) 
 \rho^{( k+r)}(\underline \xi^{( k+r)} )\big )  \text{ }\mu_{\omega}(\underline{\xi}^{( k+r)}).
\end{align} 

The error probability for the ``non-demolition model'' is denoted by $\tilde \epsilon^{(k,k+r)}(\nu)$.
In the context of the ``non-demolition model'', we are facing a classical parameter estimation problem that is equivalent to a problem concerning multinomial tests \cite{Hoef} and hence can be solved using large deviation inequalities, as in Theorem \ref{bs}. In particular, it is possible to characterize the speed of convergence to the values $\nu$ on the level sets of the estimator, for the ``non-demolition model''. 

The hypotheses underlying the results in this subsection are the ones summarized in Assumption \ref{Ass:m1}.
Our first result is the  following lemma.
 
 \begin{lemma}[Estimation Fidelity]
\label{lem:m2}
There are constants $a \in (0,1)$  and $C>0$ such that,  for all $k, r \in \mathbb{N}$, the following bounds  hold uniformly in $\rho^{(0)}$:   
\begin{align}\label{abu1}
|  \tilde \mu_{\omega}(\Xi_{\hat{\mathcal{N}}^{(k,k+r)} =\nu}) - \tilde{ \mathbb{E}} \mathrm{Tr}(\Pi_{\nu} \tilde \rho^{(k)}) | & \leq 2 C a^{r},\\
\Big |\mu_{\omega}(\Xi_{\hat{\mathcal{N}}^{(k,k+r)} =\nu}) -  \mathbb{E}  \mathrm{Tr}(\Pi_{\nu} \rho^{(k)}) \Big | & \leq 2 C a^{r}  +  \frac{d_1}{d^{-1}-1} d^{-r-1},\label{otra01}
\end{align}
and 
\begin{align}\label{abu2}
\tilde \epsilon^{(k,k+r)}(\nu) &\leq C a^{r},\\ 
\epsilon^{(k,k+r)}(\nu) &\leq  C a^{r}+ \frac{ d_1 }{d^{-1}-1} d^{-r-1}, \label{otra02}
\end{align}
for all $\nu \in \Xi_{\infty}$.
\end{lemma}

\noindent The proof of Lemma \ref{lem:m2} is not particularly difficult but  lengthy, and we  defer it to the appendix; (see Section \ref{App2}). We note that, in the ``non-demolition model'', 
$$\tilde{ \mathbb{E}} \mathrm{Tr}(\Pi_{\nu} \tilde \rho^{(k)})=\mathrm{Tr}(\Pi_{\nu} \rho^{(0)}),$$
 and hence $\tilde{\mathbb{E}} \mathrm{Tr}(\Pi_{\nu} \tilde \rho^{(k)})$ can be replaced by $\mathrm{Tr}(\Pi_{\nu}  \rho^{(0)})$ in \eqref{abu1}. The basic ideas of the proof of Lemma \ref{lem:m2} are quite straightforward. First, one controls the ``non-demolition dynamics'' using large deviation estimates that lead to \eqref{abu1} and \eqref{abu2}. Then one controls the difference between the true dynamics, as given by the composition $\Phi^{(k)}_{\underline{\xi}} \circ \Phi^{(k-1)}_{\underline{\xi}} \circ \cdot\cdot\cdot$ of completely positive maps, and the ``non-demolition dynamics'', using simple perturbative estimates. This leads to the bounds given in \eqref{otra01} and \eqref{otra02}. The constant $a$ turns out to be directly related to the infimum over $\nu \neq \nu'$ of the relative entropies $I_{p_{\nu}}(p_{\nu'})$; see  Section \ref{App2}. 
\vspace{2mm}

We deduce from the definition of the error probability $ \epsilon^{(k,k+r)}(\nu) $ that    
\begin{align}
\label{eq:m7}
 \mu_{\omega} \Big ( \Big \{ \underline \xi \:\Big | \:  \min_{\nu \in \Xi_{\infty}} \tr( (1- \Pi_{\nu}) \rho^{(k+r)}(\underline\xi )) \geq \Delta \Big \} \cap 
\Xi_{\hat{\mathcal{N}}^{(k,k+r)} =\nu}  \Big ) \leq  \frac{  \epsilon^{(k,k+r)}(\nu)   }{\Delta},
\end{align}
for  any  $\Delta>0$.
Since, by definition, the union of the sets $ \Xi_{\hat{\mathcal{N}}^{(k,k+r)} =\nu} $
is the entire space $\Xi$, we conclude that 
\begin{align} 
\label{eq:m7prima}
  \mu_{\omega}  \Big ( \Big \{ \underline \xi \:\Big | \:  \min_{\nu \in  \Xi_{\infty} } \tr( (1- \Pi_{\nu})  \rho^{(k+r)}(\underline\xi ))  \geq  \Delta \Big \}  \Big ) \leq   \sum_{\nu \in \Xi_{\infty}} \frac{  \epsilon^{(k,k+r)}(\nu)   }{\Delta}, 
\end{align}
 uniformly in $\rho^{(0)}$. Combining this with Eq.\eqref{otra02}, we obtain the first part of our main result, Theorem \ref{princ1}.

\begin{theorem}[Jump process]\label{princ1}
Let $\varepsilon \in (0, 1]$. If $r$ is large enough and if $d_1$ is small enough,  then
\begin{align}
\label{eq:m7primaprimaprima}
  \mu_{\omega}  \Big (  \Big \{ \underline \xi \:\Big | \:  \max_{\nu \in \Xi_{\infty}} \mathrm{Tr}(  \Pi_{\nu} \rho^{(k+r)}(\underline\xi ))  \geq  1-  \varepsilon \Big \}  \Big )  \geq  1 - \varepsilon,
\end{align}
 for all $k \geq 0$, uniformly with respect to the initial condition  $\rho^{(0)}$. Furthermore, 
\begin{equation}\label{fin}
  \mu_{\omega}  \Big ( \Big \{  \underline \xi \hspace{.2cm} \Big | 
\hspace{.3cm}  \exists \: \nu  \in  \Xi_{\infty}\: : \:    \|  \rho^{(k+r)}(\underline\xi ) - \Pi_{\nu} \rho^{(k+r)}(\underline\xi ) \Pi_{\nu} \| \leq  \varepsilon \Big \}  \Big )\geq 1 - \varepsilon, 
\end{equation} 
uniformly with respect to  $\rho^{(0)}$. 
\end{theorem}

Theorem \ref{princ1} says that if the dynamics of the system $S$ is very close to a ``non-demolition dynamics'' then the evolution of the reduced state of $\overline{P}$ is very close to one described by a stochastic process  on the space of density matrices commuting with the operator 
$\mathcal{N}$, uniformly in the initial state $\rho^{(0)}$. In other words, the dynamics of the reduced states of $\overline{P}$ is encoded into a jump process on the spectrum, $\Xi_{\infty}$, of $\mathcal{N}$. Unfortunately, we are unable to determine the transition rates  of this jump process for general models. However, in certain limiting regimes, it is given by a Markov chain with explicit transition probabilities. 
Figure 2, below, refers to  a simple model of this kind. 

\begin{figure}[H]
\begin{center}
\begin{tikzpicture}
 \def\a{0.125};
  \def\b{8};

  \fill[color=red!30] (1.5,-0.1) rectangle (8.5,0.1);  
       \draw[->,thick] (0,0)--(14,0.0);

              \draw[<-] (0.5,-0.8)--(0.8,-0.8); 
              \draw[->] (1.2,-0.8)--(1.5,-0.8);
            
            \draw[-,color=blue] (0.5,-0.5)--(0.5,0.5);
              \draw[-,color=blue] (0.5+\a,-0.5)--(0.5+\a,0.5);
             \draw[-,color=blue] (0.75,-0.5)--(0.75,0.5);
             \draw[-,color=blue] (0.75+\a,-0.5)--(0.75+\a,0.5);
             \draw[-,color=blue] (1,-0.5)--(1,0.5);
                \draw[-,color=blue] (1+\a,-0.5)--(1+\a,0.5);
              \draw[-,color=blue] (1.25,-0.5)--(1.25,0.5);
              \draw[-,color=blue] (1.25+\a,-0.5)--(1.25+\a,0.5);
               \draw[-,color=blue] (1.5,-0.5)--(1.5,0.5);
               
            \draw[color=blue] (1,-0.7) node[below] {  \footnotesize $\lambda_1$}; 
            
            \draw[color=blue,decorate,decoration={brace}]
  (0.5,0.7) -- (1.5,0.7) node[above,pos=0.5] {M};

                \draw[<-] (0.5+\b,-0.8)--(0.8+\b,-0.8); 
              \draw[->] (1.2+\b,-0.8)--(1.5+\b,-0.8);
            
            \draw[-,color=blue] (0.5+\b,-0.5)--(0.5+\b,0.5);
              \draw[-,color=blue] (0.5+\a+\b,-0.5)--(0.5+\a+\b,0.5);
             \draw[-,color=blue] (0.75+\b,-0.5)--(0.75+\b,0.5);
             \draw[-,color=blue] (0.75+\a+\b,-0.5)--(0.75+\a+\b,0.5);
             \draw[-,color=blue] (1+\b,-0.5)--(1+\b,0.5);
                \draw[-,color=blue] (1+\a+\b,-0.5)--(1+\a+\b,0.5);
              \draw[-,color=blue] (1.25+\b,-0.5)--(1.25+\b,0.5);
              \draw[-,color=blue] (1.25+\a+\b,-0.5)--(1.25+\a+\b,0.5);
               \draw[-,color=blue] (1.5+\b,-0.5)--(1.5+\b,0.5);
               
            \draw[color=blue] (1+\b,-0.7) node[below] {  \footnotesize $\lambda_1$}; 
                               
                 \draw[<->] (0.5,-1.2)--(8.5,-1.2); 
                 \draw (5,-1.2) node[below] {  \footnotesize 1 cycle of duration $\lambda_2$};

            \draw[color=blue,decorate,decoration={brace}]
  (8.5,0.7) -- (9.5,0.7) node[above,pos=0.5] {M};

\end{tikzpicture}
\end{center}
\end{figure}

In Figure 2, the duration, $\lambda_2$, of a full measurement cycle is much larger than the length, $\lambda_1$, of the time interval during which projective measurements in $E$ are carried out. In time intervals of length $\lambda_2 - \lambda_1$, the evolution of the system $S$ is unitary; the corresponding unitary propagator being given by 
$$e^{-i (\lambda_2- \lambda_1)H_S},$$
for some Hamiltonian $H_S$ acting on $\mathcal{H}_S = \mathcal{H}_{\overline{P}} \otimes \mathcal{H}_E$ of the form 
$$H_S = H_{\overline{P}} \otimes \id + \id \otimes H_E.$$
In time intervals of length $\lambda_1$, the unitary evolution is ``interrupted'' by $M$ projective measurements of a quantity of $E$ represented by a self-adjoint operator $X$ acting on $\mathcal{H}_E$; (we refer to the Introduction for a simple, concrete model; and to \cite{FS} for an outline of the theory of projective measurements). The times of measurement of $X$, during the $(n+1)^{st}$ cycle, are approximately given by
\begin{equation*}
n \lambda_2  +   \frac{j}{M-1}\lambda_1, \qquad j=0,1,...,M-1.
\end{equation*}
If $M$ is sufficiently large and $\lambda_1$ is small enough,  we can use Theorem  \ref{princ1} to approximately calculate the transition probabilities for jumps on the state space $\Xi_{\infty} = \lbrace 1,...,N \rbrace$ during each cycle. It is not hard to show that, in the limit where first $M \rightarrow \infty$ and then 
$\lambda_1 \rightarrow 0$ the transition probabilities for jumps from $\nu$  to $\nu' $ (with $\nu, \nu' \in \Xi_{\infty}$) approach the ones of a Markov chain whose transition probabilities are given by
\begin{equation}
\mathrm{Tr}(\Pi_{\nu} e^{i \lambda_2 H_{\overline{P}}} \Pi_{\nu'}  e^{-i \lambda_2 H_{\overline{P}}} \Pi_{\nu}).
\end{equation}

To conclude this section, we present the proof of the second inequality in Theorem \ref{princ1}. 

\begin{proof} \textit{(Second inequality of Theorem \ref{princ1})} 
We first remark that 
\begin{align*}  
 \mathbb{E}  &   \| \Pi_\nu   \rho^{(k+r)} \Pi_{\nu'}   \|   =  \sum_{\underline{\xi}^{(k+r)} }  
  \| \Pi_\nu   \rho^{(k+r)}(\underline{\xi}^{(k+r)}) \Pi_{\nu'}   \|  \text{ } \mu_{\omega}(\underline{\xi}^{(k+r)})\\
&=  \sum_{\underline{\xi}^{(k+r)} }  
  \| \Pi_\nu   \Phi_{ \underline{\xi}}^{(r+ k)} \circ \cdots  \circ \Phi_{ \underline{\xi}}^{(k+1)} [  \rho^{(k)} (\underline{\xi}^{(k)}) ]   \Pi_{\nu'}   \|  \text{ }\mu_{\omega}(\underline{\xi}^{(k)})\\
  & \leq ( \delta_{\nu \nu'})^{r} +    \frac{ d_1 }{d^{-1}-1}  d^{-r-1}, 
\end{align*}
where we  have used \eqref{dec} and a simple perturbative estimate based on item (i) of Assumption \ref{Ass:m1}; (see also Appendix \ref{App2} for  similar, more detailed calculations). The above inequality and \eqref{eq:m7primaprimaprima} enable us to choose $r$ so large and  $d_1$ so small (depending on $ r$ and $d_2$) that 
\begin{align} \label{nop1}
\sum_{\nu \ne \nu'}  \mathbb{E}   \Big [ \| \Pi_\nu \rho^{(k+r)} \Pi_{\nu'}   \|   \Big ] \leq  \varepsilon^2/4,  
\end{align}
and such that \eqref{eq:m7primaprimaprima} is fulfilled, with $\varepsilon$ replaced by $\varepsilon/2$. Introducing the sets
\begin{equation}\label{put1}
\mathcal{A}_{k+ r} :=  \Big \{ \underline \xi \:\Big | \:  \max_{\nu \in \Xi_{\infty}} \tr(  \Pi_\nu \rho^{(k+r)}(\underline\xi )) \leq 1-  \varepsilon/2 \Big \},  
\end{equation}
and
\begin{equation}\label{put2}
\mathcal{B}_{k+r} :=  \Big \{ \underline \xi \: \: \Big | \: \:  \sum_{\nu \ne \nu'} \:   \| \Pi_\nu  \rho^{(k+r)}(\underline \xi) \Pi_{\nu'}   \|    \geq  \varepsilon/2  \Big \}, 
\end{equation}
for all $k\geq 0$, it follows from  \eqref{eq:m7primaprimaprima}  and  \eqref{nop1} that $ \mu_{\omega}(\mathcal{A}_{k+ r} ) \leq  \varepsilon/2$ and  $\mu_{\omega}(\mathcal{B}_{k+r} ) \leq  \varepsilon/2$, for all $k \geq 0$.  Let $\underline \xi \not\in \mathcal{A}_{k+r} \cup \mathcal{B}_{k+ r}$. Eq. \eqref{put1} shows that there exists some $\nu_{0} \in \Xi_{\infty}$ such that
\begin{equation}\label{casi1}
\sum_{\nu \ne \nu_0} \tr(  \Pi_\nu  \rho^{(k+r)}(\underline\xi )) =
\sum_{\nu \ne \nu_0} \|  \Pi_\nu  \rho^{(k+r)} \Pi_\nu  \| \leq \varepsilon/2,  
\end{equation}
while Eq. \eqref{put2} implies that 
\begin{equation}\label{casi2}
 \sum_{\nu \ne \nu'} \:   \| \Pi_\nu  \rho^{(k+r)}(\underline \xi) \Pi_{\nu'}   \|    \leq  \varepsilon/2. 
\end{equation}
Gathering \eqref{casi1}-\eqref{casi2},  we finally deduce that  $
\|  \rho^{(k+r)}(\underline \xi)  - \Pi_{\nu_0}  \rho^{(k+r)}(\underline \xi) \Pi_{\nu_0} \| \leq \varepsilon.   
$
\end{proof}

\section{Appendix }
In this appendix we present all proofs that have been omitted in the main body of the text.
\subsection{Proofs of results in Section \ref{Section3}} \label{App1}
We begin this subsection by explaining why the measures $\mu_{\omega}(\cdot \vert \nu)$ and $\mu_{\omega}(\cdot \vert \nu')$ are mutually singular, unless $\nu=\nu'$. Afterwards, we complete the proof of Lemma \ref{312}.
\subsubsection{Comments concerning the mutual singularity of the product measures $\mu_{\omega}(\cdot \vert \nu)$}

Since the measures $p_{\nu} \neq p_{\nu'}$ unless $\nu=\nu'$,  we choose  open sets $\mathcal{U}_\nu$ (defined by the metric on $\mathcal{M}(\sigma_S)$)  centered on $p_{\nu}$ such that $\mathcal{U}_\nu \cap \mathcal{U}_{\nu'}=\emptyset$ unless $\nu=\nu'$. For every $ \nu \in \Xi_{\infty}$ we define $ \Xi_\nu^{(n)} \subset \Xi_{\infty} $ as the inverse image of the set $\mathcal{U}_\nu$ under  the empirical measure $\sum_{\xi \in \sigma_S} f_{\xi}^{(n)} \delta_{\xi}$ .
 It directly follows from \eqref{largee} that 
 \begin{equation}\label{es1}
\mu_{\omega}(\Xi_{ \nu'}^{(n)} \vert \nu) \leq\mu_{\omega}((\Xi_\nu^{(n)})^c \vert \nu ) \leq C_{ \mathcal{U}_\nu } \exp(-n  I_{p_\nu}( (\mathcal{U}_\nu)^c )/2)   
\end{equation}
if $\nu' \neq \nu$, and  hence Borel-Cantelli's Lemma implies that 
$$\mu_{\omega} (\limsup  (\Xi_\nu^{(n)})^c \vert \nu) = 0 = \mu_{\omega}(\limsup  (\Xi_{ \nu'}^{(n)}) \vert \nu) = 0    . $$ 
We define 
$ \Xi_\nu  =  \Big ( \limsup  (\Xi_\nu^{(n)})^c \Big )^c =
 \liminf  (\Xi_\nu^{(n)})
 $. Then we have that 
\begin{equation}\label{es2}
\mu_{\omega}( \Xi_{\nu'} \vert \nu ) = \delta_{\nu,\nu'}.
 \end{equation}
Clearly, 
\begin{equation} \label{es3}
\mathbb{E}  \Big ( \|  \Pi_{\nu} \rho^{(n)} \Pi_{\nu} \| 
\chi_{ \Xi_{ \nu'}^{(n)}  } \Big ) = \mu_{\omega}( \Xi_{ \nu'}^{(n)} \vert \nu ) \text{ }
\| \Pi_{\nu} \rho^{(0)} \Pi_{\nu} \|,
\end{equation}
where $ \mathbb{E} $ denotes the expectation value with respect to the measure $\mu_{\omega}$.  Using again \eqref{es1}, we deduce that 
\begin{equation}
\sum_{n \geq 0} \mathbb{E}  \Big ( \|  \Pi_{\nu} \rho^{(n)} \Pi_{\nu} \| 
\chi_{ \Xi_{ \nu'}^{(n)}  } \Big )< \infty,
\end{equation}
and hence  that $\|  \Pi_{\nu} \rho^{(n)} \Pi_{\nu} \| 
\chi_{ \Xi_{ \nu'}^{(n)}  }$  converges to zero, $\mu_{\omega}$-a.e.. For  every element $\underline{\xi} \in \Xi_{ \nu'}$, there is $n_0 \in \mathbb{N}$ such that $\underline{\xi} \in \Xi_{ \nu'}^{(n)}$ for all $n \geq n_0$, and we conclude that $\|  \Pi_{\nu} \rho^{(n)} \Pi_{\nu} \| 
\chi_{ \Xi_{ \nu'}  }$  converges to zero, $\mu_{\omega}$-a.e..

\subsubsection{ Completion of the proof of Theorem  \ref{312}} \label{proof1}
As $ \sigma_{S}$ is finite and the relative entropy function $ I_{p_\nu} $ is lower-semicontinuous,  we can choose the sets $\{  \mathcal{U}_{\nu} \}_{\nu \in  \Xi_{\infty}}$ (also used in the previous paragraph) such that they satisfy:
\vspace{2mm}

 (a) $  \mathcal{U}_\nu   \cap \mathcal{U}_{\nu'}  = \emptyset$, 

 (b) For every $\nu \ne \nu' : $
$ I_{p_\nu}(p) \geq \min_{\nu \neq \nu'}  I_{p_\nu}(p_{\nu'})  - \delta/2 
\qquad \text{ }\forall \: p \in  \mathcal{U}_{\nu'}   $,
\vspace{2mm}

\noindent where $0<\delta \leq  \min_{\nu \neq \nu'}  I_{p_\nu}(p_{\nu'}) $.   In this particular case, it is  convenient to identify $\mathcal{M}(\sigma_S)$ with the compact convex subset of $\mathbb{R}^{\vert \sigma_S \vert}$ made of all vectors $p=(p_1,...,p_{\vert \sigma_S \vert})$ such that $p_i \geq 0$ and $\sum_i p_i=1$. We define again the sets $\Xi_{\nu}^{(n)}$  as the inverse image of  the sets $\mathcal{U}_\nu$ under the empirical measure $\sum_{\xi \in \sigma_S} f_{\xi}^{(n)} \delta_{\xi}$. 
Since $\Xi_{\infty}$ is finite, we deduce from Theorem \ref{bs} that there are  constants  $ a \in (0, 1) $ and $C > 0$ such that,  for every $\nu \ne \nu'$,
\begin{equation}\label{deci1sn}
\mu_{\omega}(\Xi_{\nu'}^{(n)} \vert \nu )\leq \mu_{\omega}((\Xi_{\nu}^{(n)})^c \vert \nu ) \leq C a^n .
\end{equation}
Using \eqref{deci1sn} we obtain 
\begin{align} \label{co1}
 \mathbb{E} \Big (  \chi_{ \Xi_{\nu}^{(n)} }  \sum_{\nu'\ne \nu}  {\rm Tr}(\Pi_{\nu'} \rho^{(0)} )  \frac{    \mu_{\omega}(\cdot \vert \nu')}{\mu_{\omega}   }  \Big ) 
 \leq C   a^n.  
\end{align}
Eq.  \eqref{co1}  and some easy calculations imply further that 
\begin{align} \label{coucoucsn}
 \mathbb{E} \Big [    
\chi_{ \Xi_{\nu}^{(n)} }  &  \Big \|  {\rm Tr}(\Pi_{\nu} \rho^{(0)})  \Pi_{\nu} \rho^{(n)}
 \Pi_{\nu} 
 -    \Pi_{\nu} \rho^{(0)} \Pi_{\nu} \Big \| \Big ]    \leq C a^n. 
\end{align}
Following similar procedures as above with 
$\chi_{ \Xi_{\nu}^{(n)} }    \Big \|  {\rm Tr}(\Pi_{\nu} 
\rho^{(0)} \Pi_{\nu})  \Pi_{\nu} \rho^{(n)}
\Pi_{\nu} 
 -    \Pi_{\nu} \rho^{(0)} \Pi_{\nu} \Big \|$ instead of 
$ \|  \Pi_{\nu} \rho^{(n)} \Pi_{\nu} \| 
\chi_{ \Xi_{\tilde \nu}^{(n)}  } $, using the fact that $ \liminf \Xi^{(n)}_\nu = \Xi_\nu $, we find that  
\begin{equation} \label{colo1}
\lim_{n \to \infty } \chi_{ \Xi_\nu}     \Big \|  {\rm Tr}(\Pi_{\nu} 
\rho^{(0)} \Pi_{\nu})  \Pi_{\nu} \rho^{(n)}
 \Pi_{\nu} 
 -    \Pi_{\nu} \rho^{(0)} \Pi_{\nu} \Big \| = 0,
\end{equation}
$\mu_{\omega}$-almost surely.

\subsection{Proofs of results in Section \ref{Section4}} \label{App2}

In this part of the appendix we prove Lemma \ref{lem:m2}. For ease of comprehension, the proof  is subdivided into a series of lemmas. The inequalities of Lemma  \ref{lem:m2} are proven in paragraphs \ref{stepone} and \ref{steptwo}.

\subsubsection{Auxiliary lemmas} \label{intermed}

Our first result concerns the speed of ``purification'' in non-demolition measurements. We recall  that the functionals 
\begin{equation} \label{petitrappel}
\tilde{\mu}_{\omega}(\xi_1,...,\xi_n \vert \nu)=\prod_{i=1}^{n} p(\xi_i \vert \nu) \qquad \text{and} \qquad \tilde{\mu}_{\omega}(\xi_1,...,\xi_n)= \sum_{\nu \in \Xi_{\infty}} \tr(\Pi_{\nu} \rho^{(0)}) \tilde{\mu}_{\omega}(\xi_1,...,\xi_n\vert \nu)
\end{equation}
generate probability measures on the space $(\Xi, \Sigma)$ that we denote by the same symbols, and that the probability distributions $p_{\nu}=\sum_{\xi \in \sigma_S} p(\xi \vert \nu) \delta_{\xi}$ and $p_{\nu'}=\sum_{\xi \in \sigma_S} p(\xi \vert \nu') \delta_{\xi}$ on $\sigma_S$ are assumed to be distinct unless $\nu=\nu'$. Furthermore it is assumed that $p(\xi \vert \nu) \neq 0$ for all $\xi, \nu$; see Item (ii) of Assumption \ref{Ass:m1}, and that the sets $\Xi_{\infty}$ and $\sigma_S$ are finite.

\begin{lemma}\label{conv}
Let $\delta \in (0, \min_{\nu \neq \nu'}  I_{p_\nu}(p_{\nu'}))$. There are  constants $C>0$ and 
$a \in (0, 1)$ (independent of $\rho^{(0)}$),  and a family of sets $\Xi^{(n)}_\nu \subset \Xi$ (independent of $\rho^{(0)}$) such that the following holds true: 
\begin{enumerate}[(i)]
\item 
$\Xi_{\nu}^{(n)} \cap \Xi_{\nu'}^{(n)}=\emptyset$ and $\vert  \tilde{\mu}_{\omega}(\Xi_{\nu'}^{(n)}\vert \nu) \vert   \leq Ca^n$   if  $\nu \neq \nu'$.
\item For   all $\nu \in  \Xi_{\infty}$,
\begin{equation}
\label{disjoi}
\vert  \tilde{\mu}_{\omega}(\Xi_{\nu}^{(n)}) - {\rm{Tr}}(\Pi_{\nu} \rho^{(0)})  \vert \leq C a^n, \qquad  
  1-  \tilde{\mu}_{\omega}\big( \underset{\nu \in \Xi_{\infty}}{ \bigcup} \Xi_{\nu}^{(n)} \big) \leq
 C | \Xi_{\infty}| a^n. 
\end{equation}

\item If $\underline{\xi} \in  \Xi_{\nu}^{(n)}$ and if $\mathrm{Tr}(\rho^{(0)} \Pi_{\nu}) > 0$, then 
\begin{equation} \label{concon}
\sum_{\nu'\neq \nu} \mathrm{Tr}(\Pi_{\nu'} \tilde \rho^{(n)}(\underline{\xi}^{(n)})) \leq \sum_{\nu'\neq \nu} \frac{ \mathrm{Tr}(\Pi_{\nu'} \rho^{(0)})}{ \mathrm{Tr}(\Pi_{\nu} \rho^{(0)}) } e^{- n ( \underset{ \nu \neq \nu'}{\min} I_{p_\nu}(p_{\nu'}) - \delta )}. 
\end{equation}
\end{enumerate}
\end{lemma}
\vspace{2mm}

\begin{proof} The assumption  that $p(\xi \vert \nu) \neq 0$ for all $\xi, \nu$ implies that the sets $\mathcal{U}_{\nu}$ constructed in Paragraph \ref{proof1} can be chosen such that the further property
\vspace{2mm}

(c) For every $\nu$,  $ I_{p_\nu}(p) \leq \frac{\delta}{2}, \text{ } \qquad \forall \: p \in  \mathcal{U}_{\nu},   $
\vspace{2mm}

\noindent is satisfied (as well as (a) and (b) of Paragraph  \ref{proof1}), because the relative entropy functions are continuous in that case. The family of sets $(\Xi_{\nu}^{(n)})$ is defined as in Section \ref{proof1} and still satisfies
\begin{equation}\label{deci0sn}
\tilde \mu_{\omega}(\Xi_{\nu'}^{(n)}\vert \nu)\leq \tilde \mu_{\omega}((\Xi_{\nu}^{(n)})^c\vert \nu) \leq C a^n, 
\end{equation}
if $\nu \neq \nu'$ and hence Item (i) is clear. Furthermore, 
\begin{align}\label{alf1}
\Big | \tilde \mu_{\omega}(\Xi_{\nu}^{(n)}   )-  {\rm Tr}(\Pi_\nu \rho^{(0)}  )  \Big |  
=  \Big |  \sum_{\nu ' \in \Xi_{\infty}} {\rm Tr}(\Pi_{\nu'} \rho^{(0)}) \tilde \mu_{\omega} (\Xi_{\nu}^{(n)} \vert \nu'  )-  {\rm Tr}(\Pi_\nu \rho^{(0)})   \Big |
\leq C  a^n,
\end{align}
where we have used that $  \sum_{\nu' } {\rm Tr}(\Pi_{\nu'} \rho^{(0)}) = 1 $. This implies the first equation in  Item (ii). The second equation of  Item (ii) follows easily:
\begin{align} \label{alf2}
\Big | 1 - \tilde \mu_{\omega}\Big ( \bigcup_{\nu \in \Xi_{\infty}}  \Xi_{\nu}^{(n)} \Big) \Big | = 
\Big | \sum_{\nu \in \Xi_{\infty}} {\rm Tr}(\Pi_\nu \rho^{(0)}  ) - 
\sum_{\nu \in \Xi_{\infty}} \tilde \mu_{\omega}\Big (   \Xi_{\nu}^{(n)} \Big)
\Big | 
\leq  C |\Xi_{\infty}| a^n,  
\end{align}  
where we use \eqref{alf1}. We finally prove Item (iii).  We recall that the empirical frequencies are given by  
\begin{equation}\label{alf3} 
f_{\xi'}^{(n)}( \underline{\xi})  = \frac{1}{n} \# \Big \{ j \in \{1, \cdots, n \} \: \Big | \: \xi_j = \xi'  \Big \}.
\end{equation}

\noindent We have that 
\begin{equation}\label{alf5.1}
\begin{split}
\exp \Big (- n I_{p_\nu} \big( \sum_{\xi' \in \sigma_S} f_{\xi'}^{(n)}( \underline{\xi}) \delta_{\xi'} \big)
\Big )& =
 \prod_{\xi' \in \sigma_S}   p(\xi' \vert  \nu )^{ n f_{\xi'}^{(n)}(\underline \xi) } 
  \prod_{\xi' \in \sigma_S}   f_{\xi'}^{(n)}(\underline \xi)^{ - n f_{\xi'}^{(n)}(\underline \xi) }  \\
  &= \tilde  \mu_{\omega} (\underline{\xi}^{(n)} \vert \nu)
  \prod_{\xi' \in \sigma_S}   f_{\xi'}^{(n)}(\underline \xi)^{ - n f_{\xi'}^{(n)}(\underline \xi) }   .   
\end{split}
\end{equation}
Using  \eqref{alf5.1}, we deduce that  for $\underline \xi \in \Xi^{(n)}_\nu $ and $\nu' \ne \nu$,   
\begin{align*} 
{\rm Tr} \Big (\Pi_{\nu'} \tilde \rho^{(n)}(\underline{\xi}^{(n)}) \Big ) = & \frac{\tilde{\mu}_{\omega} (\underline{\xi}^{(n)} \vert \nu')}{ \tilde \mu_{\omega}(\underline{\xi}^{(n)}) } \text{ }
{\rm Tr}(\Pi_{\nu'} \rho^{(0)} ) \leq  \frac{\tilde{\mu}_{\omega}(\underline{\xi}^{(n)} \vert \nu' )}{ \tilde  \mu_{\omega}(\underline{\xi}^{(n)}\vert \nu)} \text{ }
\frac{{\rm Tr}(\Pi_{\nu'} \rho^{(0)} )}{  {\rm Tr}(\Pi_{\nu} \rho^{(0)} ) } \frac{  \prod_\xi   f_\xi^{(n)}(\underline \xi)^{ - n f_\xi^{(n)}(\underline \xi) }   }{   \prod_\xi   f_\xi^{(n)}(\underline \xi)^{ - n f_\xi^{(n)}(\underline \xi) }   }
\\ \notag \leq  & \frac{{\rm Tr}(\Pi_{\nu'} \rho^{(0)} )}{  {\rm Tr}(\Pi_{\nu} \rho^{(0)} ) }
\exp \Big ( n I_{p_\nu} \big( \sum_{\xi' \in \sigma_S} f_{\xi'}^{(n)}( \underline{\xi}) \delta_{\xi'} \big) - n I_{p_{\nu'}} \big( \sum_{\xi' \in \sigma_S} f_{\xi'}^{(n)}( \underline{\xi}) \delta_{\xi'} \big) \Big  ). 
\end{align*}
As $ \underline \xi   \in \Xi^{(n)}_\nu $,  it follows from the property (c) that
\begin{equation}\label{alf7}
{\rm Tr} \Big (\Pi_{\nu'} \tilde \rho^{(n)}(\underline{\xi}^{(n)}) \Big )  \leq  \frac{{\rm Tr}(\Pi_{\nu'} \rho^{(0)} )}{  {\rm Tr}(\Pi_{\nu} \rho^{(0)} ) } e^{-n ( \underset{\nu \neq \nu'}{\min} I_{p_\nu}(p_{\nu'}) - \delta )}.
\end{equation}    
\end{proof}

 Then we  compare the density matrices $\rho^{(k)}$ and   and their associated measures to density matrices obeying a non-demolition evolution. Given a density matrix $\hat \rho$ on $\mathcal{H}_{\overline{P}}$ and  a finite length protocol $\underline{\eta}^{(k)}   \in \sigma_{S}^{ \times k}$, we introduce the random variable 
\begin{equation}
\hat{\rho}(\underline{\xi}^{(r)} \vert  \underline{\eta}^{(k)} ):= \Phi^{(k+r)}_{\underline{\eta:\xi}_{r}} \circ ... \circ \Phi^{(k+1)}_{\underline{\eta:\xi}_{1}}  [ \hat \rho],
\end{equation}
 on $\sigma_{S}^{\times r}$, where $\underline{\eta:\xi}_l:=(\eta_{1},..., \eta_{k}, \xi_1,...,\xi_{l})$. We stress  that $\hat{\rho}(\underline{\xi}^{(r)} \vert  \underline{\eta}^{(k)} )$ is not normalized, and hence it is not a density matrix.

\begin{lemma}\label{comparison}
Let $r,k \in \mathbb{N}$.  For any set $\Gamma \subset \sigma_{S}^{\times r}$  and  for  any  $\underline{\eta}^{(k)} \in  \sigma_{S}^{\times k}$, we have that
\begin{align}
\| \hat{\rho}(\underline{\xi}^{(r)} \vert  \underline{\eta}^{(k)} )  - \tilde{\Phi}_{\xi_{r}} \circ ...\circ \tilde{\Phi}_{\xi_{1}} [\hat{\rho}]\| & \leq   \frac{d_1}{d^{-1}-1} d^{-r-1} \text{ }\mathrm{Tr}( \hat{\rho}(\underline{\xi}^{(r)} \vert  \underline{\eta}^{(k)} )),\\
\Big |   \sum_{\underline{\xi}^{(r)} \in \Gamma}  \Big( \mathrm{Tr}( \hat{\rho}(\underline{\xi}^{(r)} \vert  \underline{\eta}^{(k)} )) - \tr( \tilde{\Phi}_{\xi_{r}} \circ ...\circ \tilde{\Phi}_{\xi_{1}} [\hat{\rho}]) \Big) \Big | & \leq   \frac{d_1}{d^{-1}-1} d^{-r-1}, \label{54}
\end{align}
where the first inequality holds for all $\underline{\xi}^{(r)} \in \sigma_{S}^{\times r}$.

\end{lemma}
{\em Proof:}   Assumption~\ref{Ass:m1}(i)  and a standard telescopic estimate imply that for any trace class operators $\rho_1, \rho_2$,
\begin{equation}\label{co}
\|  \Phi^{(n)}_{\underline{\xi}^{(n)}} \rho_1 - \tilde{ \Phi}_{ \xi_n} \rho_2 \| \leq \| \tilde \Phi_{ \xi_n} \| \left( d_1 \| \rho_1 \| + \| \rho_1 - \rho_2 \|  \right).
\end{equation}
Iterating this inequality, we  get  that for any density matrix $\hat{\rho}$,
\begin{align}
\|   \hat{\rho}(\underline{\xi}^{(r)} \vert  \underline{\eta}^{(k)} )  -   \tilde{\Phi}_{\xi_{r}} \circ \dots \circ \tilde{\Phi}_{\xi_{1}} [\hat{ \rho}]  \| \leq   \frac{d_1}{d^{-1}-1} d^{-r-1} \text{ } \tr(  \hat{\rho}(\underline{\xi}^{(r)} \vert  \underline{\eta}^{(k)} )) \label{eq:m2}
\end{align}
where we have used that  $\| \tilde \Phi_{ \xi}\rho \|  = \tr(\tilde \Phi_{ \xi} \rho)$ for all density matrices $\rho$, and that 
\begin{equation}\label{rep}
\|  \tilde \Phi_{ \xi_{k+r} }  \|  \dots \|  \tilde \Phi_{ \xi_{k+r-m+1}  }  \| \tr(  \Phi^{(k+r-m)}_{\underline{\eta:\xi}_{r-m}} \circ ... \circ \Phi^{(k+1)}_{\underline{\eta:\xi}_{1}}  [ \hat \rho]) \leq d^{-m} \tr(  \hat{\rho}(\underline{\xi}^{(r)} \vert  \underline{\eta}^{(k)} )),
\end{equation}
which follows directly from repeated use of \eqref{eq:m1} and \eqref{co}. To prove \eqref{54}, it is sufficient to use that 
\begin{align*}
\big |  \mathrm{Tr}( \hat{\rho}(\underline{\xi}^{(r)} \vert  \underline{\eta}^{(k)} ))-  \tr( \tilde{\Phi}_{\xi_{r}} \circ ...\circ \tilde{\Phi}_{\xi_{1}} [\hat{\rho}])  \big |  \leq  \|   \hat{\rho}(\underline{\xi}^{(r)} \vert  \underline{\eta}^{(k)} )  -   \tilde{\Phi}_{\xi_{r}} \circ \dots \circ \tilde{\Phi}_{\xi_{1}} [\hat{ \rho}]  \|
\end{align*}
and to sum over  $\underline{\xi}^{(r)} \in \Gamma$. \qed
\vspace{3mm}

\subsubsection{Proof of the inequalities  \eqref{abu1} and \eqref{abu2} of Lemma \ref{lem:m2} } \label{stepone}
We use the family of  sets $\Xi^{(n)}_\nu \subset \Xi$  defined in Lemmata \ref{conv}.   
We introduce the family of sets $$\Xi^{(k,k+r)}_{\nu} : = \sigma_{S}^{\times k} \times \Xi^{(r)}_\nu  \subset \Xi.$$ The projection on the first $l$ entries of an element $\underline{\xi}$ is denoted by $\mathrm{pr}_{l}$:  $\mathrm{pr}_{l} (\underline{\xi})=\underline{\xi}^{(l)}$.  It clearly holds that
\begin{equation*} \begin{split} 
\tilde{\mu}&_{\omega}(\Xi^{(k,k+r)}_{\nu}) -  \tilde{ \mathbb{E}} \mathrm{Tr}(\Pi_{\nu} \tilde \rho^{(k)}) \\
&=  \sum_{\underline{\eta}^{(k)} } \Big(  \sum_{\underline{\xi}^{(r)} \in  \mathrm{pr}_{r}( \Xi^{(r)}_\nu)}\tr( \tilde{\Phi}_{\xi_{r}} \circ \dots \circ \tilde{\Phi}_{\xi_{1}} [\tilde \rho^{(k)}(\underline{\eta}^{(k)} )])   -  
  \tr(\Pi_\nu \tilde  \rho^{(k )}(\underline{\eta}^{(k)} )  )   \Big )
\tilde{\mu}_{\omega} (\underline{\eta}^{(k)} ). 
\end{split}
\end{equation*}
Using   Eq.~(\ref{disjoi}) of Lemma~\ref{conv} we deduce that
\begin{equation}\label{pqno1}
\big | \tilde{\mu}_{\omega}(\Xi^{(k,k+r)}_{\nu}) - \tilde{ \mathbb{E}} \mathrm{Tr}(\Pi_{\nu} \tilde \rho^{(k)}) \big |  
 \leq   Ca^{r}. 
\end{equation}
Now we notice that for sufficiently small $\delta$, $$\Xi^{(k,k+r)}_{\nu} \subset  \Xi_{\hat{\mathcal{N}}^{(k,k+r)} =\nu} \qquad \text{and} \qquad  
    \Xi_{\hat{\mathcal{N}}^{(k,k+r)} =\nu}\subset \big (\Xi^{(k,k+r)}_{\nu'}\big )^c \qquad  (\nu' \ne \nu ),$$ as follows from the construction of the sets $\Xi^{(k,k+r)}_{\nu}$; see Section \ref{intermed}. From Eq.  \eqref{deci0sn} using \eqref{petitrappel}, arguing as above, we obtain that
\begin{equation}\label{pqno2}
\tilde{\mu}_{\omega}\big ( \Xi_{\hat{\mathcal{N}}^{(k,k+r)} =\nu}\setminus   \Xi^{(k,k+r)}_{\nu} \big )  \leq  Ca^{r}.  
\end{equation}     
Eqs. \eqref{pqno1}-\eqref{pqno2} imply \eqref{abu1}.  We now prove  \eqref{abu2}. By definition of $\tilde \epsilon^{(k,k+r)}(\nu) $ and from \eqref{petitrappel}, we deduce that 
\begin{align*}
\tilde \epsilon^{(k,k+r)}(\nu) & =  \sum_{\underline{\xi} \in \Xi_{\hat{\mathcal{N}}^{(k,k+r)} =\nu}} \mathrm{Tr} \big ((1-\Pi_{\nu}) 
 \tilde{\rho}^{(k+r)}(\underline \xi^{(k+r)} )\big )  \text{ }\tilde{\mu}_{\omega}(\underline{\xi}^{(k+r)})\\
 & = \sum_{\nu' \neq \nu} \text{ } \sum_{\underline{\xi} \in \Xi_{\hat{\mathcal{N}}^{(k,k+r)} =\nu}}  \tr(\Pi_{\nu'} \rho^{(0)})  \text{ }\tilde{\mu}_{\omega}(\underline \xi^{(k+r)} \vert \nu' ). 
\end{align*}
 Eq. \eqref{abu2} follows by using  that $     \Xi_{\hat{\mathcal{N}}^{(k,k+r)} =\nu}\subset \big (\Xi^{(k,k+r)}_{\nu'}\big )^c$ if $\nu \neq \nu'$  and Eq. \eqref{deci0sn}.

\subsubsection{Proof of the inequalities \eqref{otra01} and \eqref{otra02}  of  Lemma \ref{lem:m2} } \label{steptwo}
We use the same notations as in Lemma \ref{comparison} with $\underline{\eta}^{(k)} =\underline{\xi}^{(k)}$ and  $\hat{\rho}=\rho^{(k)}(\underline{\xi}^{(k)})$. Furthermore, we use the notation $\underline{\xi}^{(k,k+r)}=(\xi_{k+1},...,\xi_{k+ r})$.  We have that
\begin{align*}
\epsilon^{(k,k+r)}(\nu) & =  \sum_{\underline{\xi} \in \Xi_{\hat{\mathcal{N}}^{(k,k+r)} =\nu}} \mathrm{Tr} \big ((1-\Pi_{\nu}) 
\rho^{(k+r)}(\underline \xi^{(k+r)} )\big )  \text{ } \mu_{\omega}(\underline{\xi}^{(k+r)})\\
&= \sum_{\underline{\xi} \in \Xi_{\hat{\mathcal{N}}^{(k,k+r)} =\nu}} \mathrm{Tr} \big ((1-\Pi_{\nu}) 
\hat{\rho}(\underline{\xi}^{(k,k+r)} \vert  \underline{\xi}^{(k)})\big )  \text{ }  \text{ } \mu_{\omega}(\underline{\xi}^{(k)})
\end{align*}
Lemma \ref{comparison} implies that
\begin{equation}\label{otra2}
\| \hat{\rho}(\underline{\xi}^{(k,k+r)} \vert  \underline{\xi}^{(k)})  - \tilde{\Phi}_{\xi_{k+r}} \circ ...\circ \tilde{\Phi}_{\xi_{k+1}} [\rho^{(k)}(\underline{\xi}^{(k)})]\|  \leq   \frac{d_1}{d^{-1}-1} d^{-r-1} \text{ }\mathrm{Tr}( \hat{\rho}(\underline{\xi}^{(k,k+r)} \vert  \underline{\xi}^{(k)} )),
\end{equation} 
and hence we deduce that
\begin{align*}
\epsilon^{(k,k+r)}(\nu) & \leq   \sum_{\underline{\xi} \in \Xi_{\hat{\mathcal{N}}^{(k,k+r)} =\nu}} \mathrm{Tr} \big ((1-\Pi_{\nu}) 
\tilde{\Phi}_{\xi_{k+r}} \circ \dots \circ \tilde{\Phi}_{\xi_{k+1}} [ \rho^{(k)}(\underline{\xi}^{(k)})] \big )  \text{ } \mu_{\omega}(\underline{\xi}^{(k)}) +   \frac{d_1}{d^{-1}-1} d^{-r-1}\\
& \leq \sum_{\nu' \neq \nu} \sum_{\underline{\xi} \in \Xi_{\hat{\mathcal{N}}^{(k,k+r)} =\nu}}  \mu_{\omega}(\underline{\xi}^{(k)}) \text{ } \mathrm{Tr} \big (\Pi_{\nu'}\rho^{(k)}(\underline{\xi}^{(k)})) \prod_{i=k+1}^{k+r} p(\xi_i \vert \nu')  +   \frac{d_1}{d^{-1}-1} d^{-r-1},
\end{align*}
and \eqref{otra02} follows by using \eqref{deci0sn} and the fact that $ \Xi_{\hat{\mathcal{N}}^{(k,k+r)} =\nu}\subset \big (\Xi^{(k,k+r)}_{\nu'}\big )^c$ if $\nu \neq \nu'$.  Eq. \eqref{otra01} follows from similar calculations and is left to the reader.

\subsection{Some properties of completely positive maps}

In this section we prove the statement left out in Section~\ref{des_non_dem}. We consider a family of completely positive maps $\Phi_{*\xi}$ acting on a finite dimensional space $\mathcal{B}(\mathcal{H})$ such that $\Phi_*(\sigma_S) = \sum_{\xi \in \sigma_S} \Phi_{*\xi}$ is a unital map. We call a collection of operators $\Gamma_\alpha$  satisfying
$$
\Phi_*(\sigma_S)[ X] = \sum_{\alpha \in I} \Gamma_\alpha^* X \Gamma_\alpha
$$ 
a Kraus decomposition of $\Phi_*(\sigma_S)$. Note that such a decomposition can always be taken as a sum of decompositions of $\Phi_{*\xi}$ over $\xi$, i.e. the index set $I$ can be written as a union of sets $I_\xi$ with $\xi \in \sigma_S$ and 
$$
\Phi_*(\sigma_S)[ X] = \sum_{\xi \in \sigma} \Phi_{*\xi} [X ]= \sum_{\xi \in \sigma} \sum_{\alpha \in I_\xi} \Gamma_\alpha^* X \Gamma_\alpha.
$$
 The following preparatory lemma is well-known (see e.g. \cite{Fagnola}). 

\begin{lemma}
Suppose that $\Phi(\sigma_S)$ has a faithful stationary state $\rho$. Then $\ker(\Phi_*(\sigma_s) - \mathds{1}) = \{\Gamma_\alpha\}_{\alpha \in I}'$ for any Kraus representation   $\Gamma_\alpha$.
\end{lemma}

{\em Proof:} We recall that $\Phi(\sigma_S) = (\Phi_*(\sigma_S))^*$ and that $\rho$ is a stationary state if $\Phi(\sigma_S) [\rho] = \rho$ and a faithful state if $\tr( \rho A^*A) = 0$ implies $A^*A=0$. We write $\Phi_* \equiv \Phi_*(\sigma_s)$.

If $A \in \ker(\Phi_* - \mathds{1})$ then so does $A^*$. Then Kadison inequality $\Phi_*[A A^*] \geq \Phi_*[A] \Phi_*[A^*] = A A^*$ holds and the difference between the LHS and the RHS can be explicitly expressed as 
$$
\Phi_*[A A^*] - A A^* = \sum_\alpha |[A^*,\Gamma_\alpha]|^2.
$$
Applying the stationary state, $\tr(\rho\, \cdot)$, on both sides we get $\tr(\rho \sum_\alpha |[A^*,\Gamma_\alpha]|^2) = 0$. Since $\rho$ is a faithful state the statement follows. \qed
\vspace{2mm}

The following lemma gives a description of a non-demolition family $\Phi_{*\xi}$.

\begin{lemma}
\label{lem:ins2}
Suppose that \begin{equation} \Phi_{* \xi} \circ  \Phi_{* \xi'} =  \Phi_{* \xi'} \circ  \Phi_{* \xi} \label{appa} \end{equation} for all $\xi,\,\xi' \in \sigma_S$ and that the map $\Phi(\sigma_S)$ has a faithful stationary state. Then $[\Phi_{* \xi}[\mathds{1}], \Phi_{* \xi'}[\mathds{1}] ] = 0$ and the corresponding joint spectral decomposition of the  identity $\mathds{1} = \sum_{\nu} \Pi_\nu$ has the property that
$$
\Phi_{*\xi} [\Pi_\nu] = p(\xi|\nu) \Pi_\nu \quad \mbox{for some} \quad p(\xi|\nu) \geq 0.
$$
Moreover if $\nu \neq \nu'$ then there exists $\xi$ such that $p(\xi|\nu) \neq p(\xi|\nu')$.
\end{lemma}

{\em Proof:} Since $\Phi_*(\sigma_S) = \sum_{\xi \in \sigma_S} \Phi_{*\xi}$, \eqref{appa} implies that  $\Phi_*(\sigma_S) \circ \Phi_{*\xi} =\Phi_{*\xi} \circ  \Phi_*(\sigma_S) $ and hence
$$
\Phi_*(\sigma_S) \circ \Phi_{*\xi}[\id] = \Phi_{*\xi}[\id].
$$
It then follows from the previous lemma that $[\Phi_{*\xi}[\id],\Gamma_\alpha] = 0$ for any Kraus decomposition $\Gamma_\alpha$ of $\Phi_*(\sigma_S)$. We then have
\begin{align*}
\Phi_{*\xi'} \circ \Phi_{*\xi}[\id] &= \sum_{\alpha \in I_{\xi'}} \Gamma_\alpha^* \Phi_{*\xi}[\id] \Gamma_\alpha
				        = \sum_{\alpha \in I_{\xi'}} \Gamma_\alpha^* \Gamma_\alpha \Phi_{*\xi}[\id] \\
				        &= \Phi_{*\xi'}[\id] \Phi_{* \xi}[\id].
\end{align*}
Repeating the same with $\xi$ and $\xi'$ exchanged we conclude that $\Phi_{*\xi'}[\id]$ and $\Phi_{*\xi}[\id]$ commute. Let now 
$\id = \sum_{\nu} \Pi_\nu$ be the joint spectral decomposition of the commuting family $\{\Phi_{*\xi}[\id]\}_{\xi \in \sigma_S}$ and suppose that
$$
\Phi_{*\xi}[\id]  = \sum_{\nu'} p(\xi|\nu') \Pi_{\nu'}. 
$$
By hitting the equation by $\Pi_{\nu}$ from the right and using  the commutativity of the operators $\Gamma_\alpha$ with $\Pi_{\nu}$ we get
$$
\Phi_{*\xi}[\Pi_\nu] = p(\xi|\nu) \Pi_\nu.
$$
To finish the proof notice that $p(\xi|\nu)$ is non-negative because $\Phi_{*\xi}$ is a positive map. Furthermore,  if $p(\xi|\nu) = p(\xi|\nu')$ for all $\xi$ then the projections $\Pi_{\nu}$ and  $\Pi_{\nu'}$ would not appear separately in the spectral decomposition, but rather their sum $\Pi_{\nu} + \Pi_{\nu'}$ would appear.
 \qed
\footnotesize
\bibliographystyle{plain}
\bibliography{BFFS_2015_final}

\begin{thebibliography}{10}

\bibitem{Aldou}
D.~Aldous.
\newblock {\em Exchangeability and related topics}.
\newblock Springer, 1985.

\bibitem{BaBeBe2}
M.~Bauer, T.~Benoist, and D.~Bernard.
\newblock Repeated quantum non-demolition measurements: convergence and
  continuous time limit.
\newblock {\em Ann. Henri Poincar{\'e}}, 14(4):639--679, 2013.

\bibitem{BaBe}
M.~Bauer and D.~Bernard.
\newblock Convergence of repeated quantum nondemolition measurements and
  wave-function collapse.
\newblock {\em Phys. Rev. A}, 84(4):044103, 2011.

\bibitem{BaBeTi}
M.~Bauer, D.~Bernard, and A.~Tilloy.
\newblock Statistics of quantum jumps and spikes, and limits of diffusive weak
  measurements.
\newblock {\em arXiv preprint arXiv:1410.7231}, 2014.

\bibitem{BePel}
T.~Benoist and C.~Pellegrini.
\newblock Large time behavior and convergence rate for quantum filters under
  standard non demolition conditions.
\newblock {\em Comm. Math. Phys.}, 331(2):703--723, 2014.

\bibitem{Chernoff}
H.~Chernoff.
\newblock A measure of asymptotic efficiency for tests of a hypothesis based on
  the sum of observations.
\newblock {\em Ann. Math. Stat.}, 23(4):493--507, 1952.

\bibitem{Finetti}
B.~De~Finetti.
\newblock La pr{\'e}vision: ses lois logiques, ses sources subjectives.
\newblock {\em Ann. Inst. Henri Poincar{\'e}}, 7(1):1--68, 1937.

\bibitem{dembo}
A.~Dembo and O.~Zeitouni.
\newblock {\em Large deviations techniques and applications}.
\newblock Springer, 2009.

\bibitem{Fagnola}
J.~Deschamps, F.~Fagnola, E.~Sasso, and V.~Umanita.
\newblock Structure of uniformly continuous quantum markov semigroups.
\newblock {\em arXiv preprint arXiv:1412.3239}, 2014.

\bibitem{Ellis}
R.~Ellis.
\newblock {\em Entropy, large deviations, and statistical mechanics}.
\newblock Springer, 2005.

\bibitem{esseen}
C.G. Esseen.
\newblock A moment inequality with an application to the central limit theorem.
\newblock {\em Scand. Actuar. J.}, 1956(2):160--170, 1956.

\bibitem{FaFS}
J.~Faupin, J.~Fr{\"o}hlich, and B.~Schubnel.
\newblock On the probabilistic nature of quantum mechanics and the notion of
  closed systems.
\newblock {\em To appear in Ann. Henri Poincar\'{e}}.

\bibitem{Feller}
W.~Feller.
\newblock {\em An introduction to probability theory and its applications},
  volume~2.
\newblock John Wiley \& Sons, 2008.

\bibitem{FS}
J.~Fr{\"o}hlich and B.~Schubnel.
\newblock Quantum probability theory and the foundations of quantum mechanics.
\newblock {\em The {M}essage of {Q}uantum {S}cience}, 2015.

\bibitem{Griffiths}
R.B. Griffiths.
\newblock {\em Consistent quantum theory}.
\newblock Cambridge Univ. Pr., 2003.

\bibitem{guerlin}
C.~Guerlin, J.~Bernu, S.~Deleglise, C.~Sayrin, S.~Gleyzes, S.~Kuhr, M.~Brune,
  J.M. Raimond, and S.~Haroche.
\newblock Progressive field-state collapse and quantum non-demolition photon
  counting.
\newblock {\em Nature}, 448(7156):889--893, 2007.

\bibitem{Hewitt}
E.~Hewitt and L.~Savage.
\newblock Symmetric measures on cartesian products.
\newblock {\em T. Am. Math. Soc.}, 80(2):470--501, 1955.

\bibitem{Hoef}
W.~Hoeffding.
\newblock Asymptotically optimal tests for multinomial distributions.
\newblock {\em Ann. Stat.}, 36(2):369--401, 1965.

\bibitem{Holevo}
A.S. Holevo.
\newblock {\em Statistical structure of quantum theory}.
\newblock Springer, 2001.

\bibitem{JaPi12}
V.~Jak{\v{s}}ic, Y.~Ogata, C.A. Pillet, and R.~Seiringer.
\newblock Quantum hypothesis testing and non-equilibrium statistical mechanics.
\newblock {\em Rev. Math. Phys}, 24(6):1230002, 2012.

\bibitem{Kraus}
K.~Kraus.
\newblock {\em States, effects and operations}.
\newblock Springer, 1983.

\bibitem{Leb14}
J.L. Lebowitz, B.~Pittel, D.~Ruelle, and E.R. Speer.
\newblock Central limit theorems, {L}ee-{Y}ang zeros, and graph-counting
  polynomials.
\newblock {\em arXiv preprint arXiv:1408.4153}, 2014.

\bibitem{Lue}
G.~L{\"u}ders.
\newblock {\"U}ber die {Z}ustands{\"a}nderung durch den {M}e{\ss}proze{\ss}.
\newblock {\em Ann. Phys.-Leipzig}, 443(5-8):322--328, 1950.

\bibitem{Mass}
H.~Maassen and B.~K{\"u}mmerer.
\newblock Purification of quantum trajectories.
\newblock {\em Lecture Notes-Monograph Series}, 48:252--261, 2006.

\bibitem{Rao}
M.M. Rao.
\newblock {\em Conditional measures and applications}.
\newblock CRC Press, 2005.

\bibitem{Schw}
J.~Schwinger.
\newblock The algebra of microscopic measurement.
\newblock {\em Proc. Natl. Acad. Sci. USA}, 45(10):1542--1553, 1959.

\bibitem{Wig}
E.P. Wigner.
\newblock {\em The {C}ollected {W}orks of {E}ugene {P}aul {W}igner}.
\newblock Springer-Verlag, 1993.

\end{thebibliography}

\end{document}